\pdfoutput=1
\documentclass[journal]{IEEEtran}

\usepackage{ifpdf}
\ifCLASSINFOpdf
  \usepackage[pdftex]{graphicx}
  \graphicspath{{figures/}}
  \DeclareGraphicsExtensions{.pdf,.png,.jpeg}
\else
  \usepackage[dvips]{graphicx}
  \DeclareGraphicsExtensions{.eps}
\fi
\usepackage{amsmath,amssymb,amsfonts}
\usepackage{amsthm}
\usepackage{booktabs}
\usepackage{array}
\usepackage{cite}
\usepackage{url}
\usepackage[caption=false,font=footnotesize]{subfig}
\usepackage{tikz}
\usetikzlibrary{arrows.meta,shapes.geometric}
\usepackage[ruled,vlined,linesnumbered]{algorithm2e}
\definecolor{tndgreen}{HTML}{086C46}
\newcommand{\pinov}[1]{\textcolor{tndgreen}{#1}}

\newtheorem{theorem}{Theorem}
\newtheorem{corollary}{Corollary}
\newtheorem{remark}{Remark}

\begin{document}

\title{Physics-Informed Unsupervised Learning for Reliable Beamforming in Cluttered MIMO-ISAC}

\author{
	Yosefine Triwidyastuti, Nghia Thinh Nguyen, and Tri~Nhu~Do
	\thanks{Y.~Triwidyastuti, N.~T.~Nguyen and T.~N.~Do are with the Telecom Neural Detection (TND) Lab, Department of Electrical Engineering,
Polytechnique Montr\'eal, Montr\'eal, QC H3T 1J4, Canada
(e-mail: \{yosefine-2.triwidyastuti, nghia-thinh.nguyen, tri-nhu.do\}@polymtl.ca).}%
}

\maketitle

\begin{abstract}
In this paper, we consider joint beamforming and power allocation for a mobile multiuser MIMO integrated sensing and communication (ISAC) network with clutter, residual self-interference, and outdated channel state information, under per-user rate and sensing signal-to-interference-plus-noise ratio (SINR) constraints that must hold on every realization. The design must be recomputed every block, and distant or weak targets can make the sensing constraint unmeetable. We propose a physics-informed unsupervised learning (PI-UnSL) network that corrects the analytical design, meets the power and norm constraints exactly by construction, and is conditioned on closed-form interference statistics, spanning the whole tradeoff in one pass. We derive a closed-form sensing SINR at the interference-aware combiner, yielding a necessary range--radar-cross-section--power condition and a priority-independent analytic violation floor, so the problem is intermittently infeasible. We then show that a dual targeting below the attainable floor has no fixed point, and calibrate the label-free objective against the measured floor. Numerical results confirm the floor and the bounded calibrated dual, and show that PI-UnSL substantially raises per-realization joint quality of service over the closed-form baseline, delivers more rate than time-division ISAC at matched detection over most of the tradeoff, and exceeds the reliability of an uncalibrated alternating-optimization solver at a fraction of its latency. More broadly, the physical knowledge given to the network, not its architecture, makes such problems learnable without labels: five architectures fail alike from scratch and perform alike when started from the closed-form design and conditioned on its interference statistics.
\end{abstract}


\section{Introduction}
\label{sec:introduction}

Integrated sensing and communication (ISAC) shares spectrum, hardware, and signal resources between wireless communication and environmental sensing~\cite{Hua_TVT_2023,Bouzabia_TVT_2025,Zhao_TSP_2026}. Combined with the spatial degrees of freedom of multiple-input multiple-output (MIMO) arrays, it can improve both functions through joint beamforming~\cite{Liu_TWC_2023,Hua_TWC_2024,Li_TVT_2025}. In multiuser MIMO-ISAC, a dual-function base station (BS) serves several users and illuminates a target on the same time--frequency resource, so the sensing waveform interferes at the users while the data streams return to the radar receiver through the target and the surrounding clutter~\cite{Liu_JPROC_2026}.
MIMO-ISAC is attractive in vehicular and other mobile networks, where the same infrastructure must deliver high-rate links while tracking the range, direction, and velocity of nearby objects~\cite{Bouzabia_TVT_2025,Menanor_TVT_2026}. There, the transmit beams, the receive combiner, and the power split must be redesigned in every coherence block, under a minimum rate for every user and a minimum sensing quality for the target~\cite{Li_TVT_2025}. Two impairments shape this design: clutter returns that compete with the target echo~\cite{Liu_JPROC_2026}, and the residual self-interference (SI) of a full-duplex transceiver that never fully cancels~\cite{Du_COMST_2026}. Moreover, a quality-of-service (QoS) requirement is specified per realization, so the rate and sensing requirements must hold on the same block rather than on average. As Table~\ref{tab:related} shows, most existing designs assume snapshot channels without clutter or SI and report mean metrics, which leaves per-realization reliability under mobility largely unaddressed.

\begin{table*}[!t]
	\centering
	\caption{Positioning Against the Closest Multiuser MIMO-ISAC Beamforming Designs}
	\label{tab:related}
	\footnotesize
	\renewcommand{\arraystretch}{1.2}
	\begin{tabular}{@{}
			>{\raggedright\arraybackslash}p{\dimexpr 0.176\textwidth-1.6\tabcolsep\relax}
			>{\raggedright\arraybackslash}p{\dimexpr 0.213\textwidth-1.6\tabcolsep\relax}
			>{\raggedright\arraybackslash}p{\dimexpr 0.218\textwidth-1.6\tabcolsep\relax}
			>{\raggedright\arraybackslash}p{\dimexpr 0.249\textwidth-1.6\tabcolsep\relax}
			>{\raggedright\arraybackslash}p{\dimexpr 0.144\textwidth-1.6\tabcolsep\relax}@{}}
		\hline
		\textbf{Work} & \textbf{Channel model / CSI} & \textbf{Design method}
		& \textbf{Constraint handling} & \textbf{Reported metric} \\
		\hline

		Iterative optimization
		\cite{Hua_TVT_2023,Hua_TWC_2024,Li_TVT_2025,Zhao_TSP_2026}
		& Snapshot, no clutter (\cite{Li_TVT_2025}: clutter as added noise), no SI, perfect CSI
		& SDR or optimal covariance (\cite{Li_TVT_2025}: fractional programming), per realization
		& Per-user SINR or CRB, power; hard constraints
		& Beampattern gain, CRB--rate; mean \\

		Model-guided learning \cite{Feng_TVT_2026}
		& Snapshot $+$ mobility, no clutter, no SI, implicit CSI$^\dagger$
		& Model-guided network, end to end
		& Power, minimum sensing SNR per target, unit-norm combiners
		& Spectral efficiency, detection; mean \\

		Label-free, expectation constrained \cite{Liu_TWC_2023}
		& Snapshot, no clutter, no SI, local CSI
		& Distributed label-free, primal--dual
		& Expectation constraints $\mathbb{E}[g]\le G$
		& Long-term mean SE \\

		Label-free, power constrained \cite{Dassanayake_WCL_2025}
		& Snapshot, no clutter, SI is negligible, perfect CSI
		& Label-free network, direct mapping
		& Power only, enforced by a projection layer
		& Rate, $P_{\rm D}$, CRB; mean \\

		\textbf{This work}
		& \textbf{Mobile tracks, clutter $+$ SI, aged CSI}
		& Label-free, analytically initialized, one pass
		& Exact normalization; dual calibrated to feasibility floor; CVaR tail
		& \textbf{Per-realization joint QoS} \\
		\hline
		\multicolumn{5}{@{}p{\textwidth}@{}}{\footnotesize
			$^\dagger$Beamformers are learned from raw pilot observations rather than from an
			estimated channel matrix.}
\end{tabular}
\end{table*}

The joint design is nonconvex, and a substantial literature solves it per realization with iterative algorithms, including semidefinite relaxation and optimal covariance structures~\cite{Hua_TVT_2023,Zhao_TSP_2026}, characterizations of the sensing--communication Pareto boundary~\cite{Hua_TWC_2024,Liu_TWC_2025}, and Lagrangian dual transforms for joint transmit--receive design~\cite{Li_TVT_2025}. These solvers deliver high-quality designs but need many iterations in every block. Under mobility, this latency is itself a performance loss. A design computed from channel state information (CSI) acquired at one block is applied to a channel that has already evolved, and performance degrades with the age of that estimate~\cite{Menanor_TVT_2026}. A design is therefore needed that is produced in a single step, before the channel it was computed for has changed.

A neural network (NN) provides such a step, replacing the per-realization optimization with one forward pass of fixed cost. Supervised training is ill-suited, since a nonconvex program supplies only locally optimal labels, at the cost of solving it for every training sample~\cite{Elbir_SPM_2023,Vaezi_Comst_2026}, which motivates unsupervised, label-free training that places the objective and the constraints directly in the loss~\cite{Liu_TWC_2023,Eisen_TSP_2019,Dassanayake_WCL_2025,Triwidyastuti_TVT_2026}. A label-free loss alone, however, does not make a constrained design learnable. Recent ISAC networks therefore retain signal-processing structure, either by unfolding an iterative solver into layers~\cite{Deka_IoT_2026} or by embedding a model-guided beamforming form~\cite{Feng_TVT_2026}. The most informative structure available is the closed-form solution that the physical model itself provides. Starting the network at that solution and supplying closed-form interference statistics as inputs leaves it only a correction to learn. Our earlier work~\cite{Triwidyastuti_TVT_2026} started training close to such a design, but only through a supervised pretraining phase. Two questions therefore remain open: whether the starting point can be reached without any labels, and whether the constrained problem is made learnable by the network architecture or by the physical knowledge it is given.

Constraint handling poses a further problem. Primal--dual and augmented-Lagrangian methods for constrained learning~\cite{Chamon_TIT_2023,Park_AAAI_2023,Eisen_TSP_2019} assume that the constraints are attainable. In MIMO-ISAC this assumption fails. For a distant target or one with a weak radar cross section (RCS), the sensing threshold is unmeetable at every power split, beam, and combiner, and a scatterer aligned with the target caps the sensing SINR at a signal-to-clutter ratio that no increase in power can raise. A dual whose target lies below such a physical limit has no fixed point, so its multiplier diverges. Mean utility conceals these failures, since it is dominated by realizations on which the constraints never bind. Attainability should therefore be measured as a channel property and calibrated for, and designs should be evaluated by how often they satisfy both requirements on the same realization.

In this paper, we propose a physics-informed unsupervised learning solution in which the NN is initialized at a closed-form analytical design (regularized zero-forcing (RZF) for the users, a target-matched beam for sensing, and a leakage-aware combiner at the radar receiver) and learns only a correction to it, rather than learning the mapping from scratch.
The distinctive point is that ``model'' here refers to the physical signal model of the channels, interference, and SINRs, from which the closed-form analytical design is derived, rather than to a statistical model of the data or to the neural network itself.
This model enters the solution three times: as the closed-form starting point, as closed-form interference statistics at the network input, and as the SINR model and the measured violation floor in the loss.
Major contributions of our work can be summarized as:
\begin{itemize}
\item We develop a mobile multiuser MIMO-ISAC model that couples spatially correlated Rician channels on moving tracks with a monostatic target, multiple clutter scatterers, residual SI, and outdated CSI, and that captures the two-way sensing--communication leakage created by superposition. On this model we formulate the joint design of the power split, transmit beams, and radar combiner under per-user rate and sensing-SINR constraints that must hold per realization. To the best of the authors' knowledge, unlike prior designs, which assume snapshot channels without clutter or SI and report mean metrics (Table~\ref{tab:related}), we evaluate designs by per-realization joint QoS.

\item At the interference-aware combiner the sensing SINR admits a closed form, from which we derive the sensing-only maximum and a necessary range--RCS--power condition. The problem is therefore intermittently infeasible, with an analytic violation floor independent of the tradeoff weight. We show that a dual whose target lies below the attainable floor, which the noise-limited bound certifies, has no fixed point, so its multiplier grows without bound. We therefore construct a label-free augmented-Lagrangian objective that charges only the gap above the measured floor, carries multipliers per priority bin, ascends the user duals on a CVaR tail, and caps the multipliers against the utility scale, which keeps the multipliers bounded.

\item We propose a physics-informed unsupervised learning (PI-UnSL) solution, a multitask network that predicts a residual correction to the RZF, target-matched, and interference-aware designs and enforces the power and norm constraints exactly by normalization. Its correction heads are initialized at zero, so the untrained network reproduces the analytical design and no labels are needed at any stage. The network is conditioned on closed-form interference statistics and on the tradeoff weight, so one network traces the full communication--sensing frontier in a single forward pass.

\item From the wireless ISAC perspective, the advantage lies in per-realization reliability rather than mean throughput. PI-UnSL substantially raises the joint-QoS probability over the closed-form baseline at comparable mean utility. It delivers more rate than time-division ISAC at matched detection over most of the tradeoff range, although time division attains a higher mean utility. It also exceeds the joint QoS of an uncalibrated alternating-optimization (AO) solver at a small fraction of its latency, while AO attains higher utility and rate. This margin persists under mobility and, reduced, under outdated CSI.

\item From the unsupervised-optimization perspective, the physical knowledge given to the network, rather than its architecture, makes the constrained problem learnable without labels. Networks of five architectures fail alike from scratch. Starting from the closed-form design provides most of the joint-QoS gain but lowers the utility, and the closed-form interference statistics and their regularizers add joint QoS and restore the utility of the baseline, after which all five architectures perform alike.
\end{itemize}

\textit{Notation:} Vectors are denoted by lowercase letters with an arrow, e.g., $\vec{x}$, and matrices by boldface uppercase letters, e.g., $\mathbf{X}$. $(\cdot)^{\mathsf T}$, $(\cdot)^{\mathsf H}$, $\operatorname{tr}(\cdot)$, $\|\cdot\|$, and $\|\cdot\|_{\mathrm F}$ denote the transpose, conjugate transpose, trace, Euclidean norm, and Frobenius norm. $\odot$ is the element-wise product, $[x]_+=\max\{0,x\}$, and $\mathcal{CN}(\vec{\mu},\boldsymbol\Sigma)$ the circularly symmetric complex Gaussian distribution with mean $\vec{\mu}$ and covariance $\boldsymbol\Sigma$. $\mathbf I_N$ is the $N\times N$ identity matrix, $\mathbb E\{\cdot\}$ the expectation, $\Re\{\cdot\}$ and $\Im\{\cdot\}$ the real and imaginary parts, and $D_{\rm KL}(\cdot\,\|\,\cdot)$ the Kullback--Leibler divergence.

\section{Mobile Multiuser MIMO-ISAC System Model}
\label{sec:system_model}

Let us consider a downlink MIMO-ISAC network with one dual-function base station (BS), $U$ single-antenna users, and one mobile sensing target. The BS has $N_{\rm t}$ transmit and $N_{\rm r}$ co-located receive antennas in a monostatic configuration, where it illuminates the target and receives the echo. Instantaneous communication channel state information (CSI) and sensing-channel parameter estimates are assumed available at the BS in each coherent processing interval (CPI). Fig.~\ref{fig:system} illustrates the network and the role of the proposed PI-UnSL, which sets the transmitter and the radar receiver of the BS in every block.

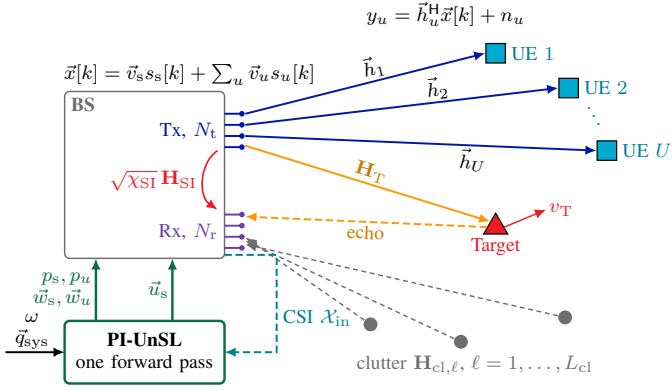
\begin{figure}[!t]
    \centering
    \resizebox{\columnwidth}{!}{\definecolor{tndblue}{HTML}{0018A8}%
\definecolor{tndpurple}{HTML}{6B3FA0}%
\definecolor{tndgreen}{HTML}{086C46}%
\definecolor{tndgray}{HTML}{6E6E6E}%
\definecolor{tndred}{HTML}{ED1B24}%
\definecolor{tndcyan}{HTML}{00A9CE}%
\definecolor{tndgold}{HTML}{D4AA00}%
\definecolor{tndorange}{HTML}{FF9800}%
\begin{tikzpicture}[
    font=\footnotesize,
    >={Latex[length=1.6mm,width=1.4mm]},
    box/.style={draw, rounded corners=2pt, align=center, line width=0.6pt, inner sep=2.5pt},
    ant/.style={line width=0.8pt},
    ue/.style={draw=black, fill=tndcyan, minimum size=2.8mm, inner sep=0pt, line width=0.5pt},
    tgt/.style={draw=black, fill=tndred, regular polygon, regular polygon sides=3,
                minimum size=4mm, inner sep=0pt, line width=0.5pt},
    clt/.style={circle, fill=tndgray, minimum size=2mm, inner sep=0pt},
    com/.style={->, line width=0.8pt, draw=tndblue},
    sen/.style={->, line width=0.8pt, draw=tndorange},
    ech/.style={->, line width=0.8pt, draw=tndorange, dash pattern=on 3pt off 1.5pt},
    cle/.style={->, line width=0.6pt, draw=tndgray, dash pattern=on 2pt off 1.5pt},
    ctl/.style={->, line width=1.0pt, draw=tndgreen},
    csi/.style={->, line width=0.8pt, draw=tndcyan!75!black, dash pattern=on 3pt off 1.5pt},
    lab/.style={inner sep=1pt},
]

\draw[line width=0.6pt, draw=tndgray, rounded corners=2pt] (0,1.9) rectangle (2.3,4.3);
\node[lab, font=\footnotesize\bfseries, text=tndgray, anchor=north west] at (0.05,4.25) {BS};
\foreach \i in {0,...,3} {
  \draw[ant, tndblue] (2.3,3.5+0.16*\i) -- (2.55,3.5+0.16*\i);
  \fill[tndblue] (2.55,3.5+0.16*\i) circle (0.04);
}
\node[lab, text=tndblue, anchor=east] at (2.25,3.74) {Tx, $N_{\rm t}$};
\foreach \i in {0,...,3} {
  \draw[ant, tndpurple] (2.3,2.05+0.16*\i) -- (2.55,2.05+0.16*\i);
  \fill[tndpurple] (2.55,2.05+0.16*\i) circle (0.04);
}
\node[lab, text=tndpurple, anchor=east] at (2.25,2.29) {Rx, $N_{\rm r}$};
\draw[->, line width=0.7pt, draw=tndred] (2.2,3.45) to[out=215, in=145] (2.2,2.6);
\node[lab, text=tndred, anchor=east] at (1.95,3.02) {$\sqrt{\chi_{\rm SI}}\,\mathbf H_{\rm SI}$};
\node[lab, anchor=south west] at (-0.05,4.35)
  {$\vec{x}[k]=\vec{v}_{\rm s}s_{\rm s}[k]+\textstyle\sum_{u}\vec{v}_u s_u[k]$};

\node[ue] (u1) at (6.2,4.85) {};
\node[ue] (u2) at (7.25,4.35) {};
\node[ue] (uU) at (7.8,3.45) {};
\node[lab, text=tndcyan!75!black, anchor=west] at ([xshift=1pt]u1.east) {UE $1$};
\node[lab, text=tndcyan!75!black, anchor=west] at ([xshift=1pt]u2.east) {UE $2$};
\node[lab, text=tndcyan!75!black, anchor=west] at ([xshift=1pt]uU.east) {UE $U$};
\node[lab, text=tndcyan!75!black, rotate=-60] at (7.62,3.92) {$\cdots$};
\draw[com] (2.6,3.98) -- (u1.west) node[pos=0.55, above, sloped, lab] {$\vec{h}_1$};
\draw[com] (2.6,3.82) -- (u2.west) node[pos=0.62, above, sloped, lab] {$\vec{h}_2$};
\draw[com] (2.6,3.66) -- (uU.west) node[pos=0.65, below, sloped, lab] {$\vec{h}_U$};
\node[lab, anchor=south west] at (4.3,5.2) {$y_u=\vec{h}_u^{\mathsf H}\vec{x}[k]+n_u$};

\node[tgt] (tg) at (6.2,2.35) {};
\node[lab, text=tndred, anchor=north] at (tg.south) {Target};
\draw[->, line width=0.6pt, draw=tndred] (tg.east) -- ++(0.6,0.25)
  node[lab, anchor=west, text=tndred] {$v_{\rm T}$};
\draw[sen] (2.6,3.52) -- (tg.north west)
  node[pos=0.5, above, sloped, lab, text=tndorange!80!black] {$\mathbf H_{\rm T}$};
\draw[ech] (tg.west) -- (2.6,2.5)
  node[pos=0.5, below, sloped, lab, text=tndorange!80!black] {echo};

\node[clt] (c1) at (4.4,0.95) {};
\node[clt] (c2) at (5.7,0.7) {};
\node[clt] (c3) at (7.2,1.05) {};
\draw[cle] (c1) -- (2.6,2.2);
\draw[cle] (c2) -- (2.6,2.15);
\draw[cle] (c3) -- (2.6,2.1);
\node[lab, text=tndgray, anchor=north] at (5.9,0.55) {clutter $\mathbf H_{{\rm cl},\ell}$, $\ell=1,\ldots,L_{\rm cl}$};

\draw[line width=1pt, draw=tndgreen, rounded corners=2pt] (0,0.1) rectangle (2.3,1.0);
\node[lab, align=center] at (1.15,0.55) {\textbf{PI-UnSL}\\one forward pass};
\draw[->, line width=0.8pt] (-0.85,0.55) -- (0,0.55)
  node[pos=0.45, above, lab, align=center] {$\omega$\\$\vec{q}_{\rm sys}$};
\draw[ctl] (0.45,1.0) -- (0.45,1.9)
  node[pos=0.5, left, lab, text=tndgreen, align=right] {$p_{\rm s},p_u$\\$\vec{w}_{\rm s},\vec{w}_u$};
\draw[ctl] (1.55,1.0) -- (1.55,1.9)
  node[pos=0.5, left, lab, text=tndgreen] {$\vec{u}_{\rm s}$};
\draw[csi] (2.3,1.95) -- (3.05,1.95) -- (3.05,0.55) -- (2.3,0.55);
\node[lab, text=tndcyan!75!black, anchor=west] at (3.1,1.05) {CSI $\mathcal X_{\rm in}$};
\end{tikzpicture}%
}
    \caption{Mobile MIMO-ISAC network with clutter and residual SI; PI-UnSL sets the beams, power split, and radar combiner in one pass.}
    \label{fig:system}
\end{figure}

Let $t$ index the coherence block of duration $T_{\rm b}$ and $k$ a sample inside it, within which ranges, angles, and radial velocities are constant so only the Doppler phase evolves. For $i\in\{1, \ldots ,U, {\rm T}\}$, the node position is $\vec{c}_i[t]=\vec{c}_i[t-1] + v_i[t] T_{\rm b} [\cos\theta_i[t],\, \sin\theta_i[t]]^{\mathsf T}$ with ground speed $v_i[t]$ and heading $\theta_i[t]$.
Random-walk, Gauss--Markov, or group-mobility laws may be assigned to $(v_i,\theta_i)$ without changing the physical-layer model.

With the BS at $\vec{c}_{\rm BS}$ and boresight as angular reference, every geometry-dependent parameter follows from range
$d_i[t] = \|\vec{c}_i[t] - \vec{c}_{\rm BS}\|_2$, angle $\phi_i[t] = \angle (\vec{c}_i[t] - \vec{c}_{\rm BS})$, and radial velocity $v_{{\rm r},i}[t] = v_i[t] \cos \vartheta_i[t]$ with $\vartheta_i[t] = \theta_i[t] - \phi_i[t]$, so $v_{{\rm r},i}>0$ on recession. For a half-wavelength uniform linear array of $N$ elements, $\vec{a}_N(\phi)=[1,\,e^{-\mathsf j\pi\sin\phi},\, \ldots,\, e^{-\mathsf j\pi(N-1)\sin\phi}]^{\mathsf T}$~\cite[Eq.~(1.23)]{Bjornson_2017}, with $\vec{a}_{\rm t} = \vec{a}_{N_{\rm t}}$, $\vec{a}_{\rm r} = \vec{a}_{N_{\rm r}}$.
The environment is represented by $L_{\rm cl}$ discrete clutter scatterers, which are passive reflectors rather than receivers and therefore enter only through two-way echo channels. We assume that these $L_{\rm cl}$ scatterers capture the dominant clutter returns and neglect weaker reflectors and diffuse clutter. This is justified by the $r^{-4}$ decay of the two-way path: a $0.01$~m$^2$ reflector at $100$~m returns about $40$~dB below the receiver noise even at full power and full array gain.

\subsection{Superimposed Sensing and Communication Signals}

The BS superimposes one dedicated sensing stream and $U$ data streams on the same block and array~\cite{Bouzabia_TVT_2025}.
Let $s_{\rm s}[k]$ denote a known unit-energy sensing symbol and $s_u[k]$ the unit-energy data symbol intended for user $u$, with $\mathbb E\{s_i[k] s_j^*[k']\} = \delta_{ij} \delta_{kk'}$ for $i,j \in \{{\rm s}, 1, \ldots, U\}$,
the streams are unit-energy, uncorrelated, and temporally white, so interference enters as covariances without cross terms between distinct delays.
For unit-norm beamforming $\vec{w}_{\rm s},\vec{w}_u \in \mathbb C^{N_{\rm t}}$ and powers $p_{\rm s},p_u$, the precoders are $\vec{v}_{\rm s} = \sqrt{p_{\rm s}} \vec{w}_{\rm s}$, $\vec{v}_u = \sqrt{p_u} \vec{w}_u$.
The transmitted MIMO-ISAC signal is $\vec{x}[k] = \vec{v}_{\rm s}s_{\rm s}[k] + \sum_{u=1}^{U} \vec{v}_u s_u[k] \in\mathbb C^{N_{\rm t} \times1}$, with covariance $\mathbf R_x = \mathbb E\{\vec{x} \vec{x}^{\mathsf H}\} = \vec{v}_{\rm s} \vec{v}_{\rm s}^{\mathsf H} + \sum_{u=1}^{U} \vec{v}_u \vec{v}_u^{\mathsf H}$ and total-power constraint $\operatorname{tr} (\mathbf R_x) = p_{\rm s} + \sum_{u=1}^{U} p_u \leq P_{\max}$.

\subsection{Spatially Correlated Rician Channel under Mobility}
\label{subsec:communication_channel}

The direct BS-to-user channel is represented by a spatially correlated Rician block-fading model. For user $u$ in block $t$,
\begin{equation}
\begin{split}
    \vec{h}_u[t]
    =\sqrt{\beta_u[t]}\Big(&
    \sqrt{K_u/(K_u+1)}\,
    e^{\mathsf j\psi_u[t]}\vec{a}_{\rm t}(\phi_u[t])\\
    &+\sqrt{1/(K_u+1)}\,
    \boldsymbol\Sigma_u^{1/2}\vec{g}_u[t]\Big),
    \label{eq:rician_channel}
\end{split}
\end{equation}
with Rician factor $K_u$~\cite{Bouzabia_TVT_2025}, line-of-sight (LoS) departure angle $\phi_u[t]$, LoS phase $\psi_u[t] = -2\pi d_u[t]/\lambda$, $\vec{g}_u[t] \sim\mathcal{CN} (\vec{0}, \mathbf I_{N_{\rm t}})$, and transmit-side correlation $\boldsymbol\Sigma_u$ of unit diagonal with $[\boldsymbol\Sigma_u]_{i,j} = \rho^{|i-j|}$, where $\rho$ is antenna spatial correlation. Since $\|\vec{a}_{\rm t}\|^2 = \operatorname{tr} (\boldsymbol\Sigma_u) = N_{\rm t}$, $\mathbb E\{\|\vec{h}_u[t]\|^2\} = \beta_u[t] N_{\rm t}$, so $\beta_u[t]$ is the per-antenna large-scale gain, where $\beta_u[t] = \beta_0 (d_u[t]/d_0)^{-\alpha_{\rm PL}} 10^{-\zeta_u[t]/10}$, with reference distance $d_0$, free-space gain $\beta_0 = (\lambda/(4\pi d_0))^2$~\cite[Eq.~(2.7)]{Goldsmith2005}, path-loss exponent $\alpha_{\rm PL}$, and $\zeta_u[t]$ is log-normal shadowing in dB. To capture spatial consistency as users move, shadowing evolves according to the Gudmundson model $\zeta_u[t] = a_u[t]\zeta_u[t-1] + \sqrt{1 - a_u^2[t]}\,\xi_u[t]$, where $\xi_u[t] \sim \mathcal{N}(0, \sigma_{\rm sh}^2)$ and $a_u[t] = \exp(-\Delta d_u[t]/d_{\rm cor})$ is the shadowing spatial correlation over the distance $\Delta d_u[t]$ traveled in block $t$~\cite[Eq.~(2.50)]{Goldsmith2005}, with decorrelation distance $d_{\rm cor} = 10$~m.

The diffuse component in \eqref{eq:rician_channel} follows the Clarke--Jakes recursion that can be expressed as
\begin{equation}
    \vec{g}_u[t]
    =\rho_u[t]\vec{g}_u[t-1]
    +\sqrt{1-\rho_u^2[t]}\,\vec{e}_u[t],
    \label{eq:gauss_markov}
\end{equation}
where $\vec{e}_u[t] \sim \mathcal{CN} (\vec{0}, \mathbf I_{N_{\rm t}})$ and $\rho_u[t] = J_0(2\pi f_{m,u}[t] T_{\rm b})$ with maximum Doppler frequency $f_{m,u}[t] = v_u[t]/\lambda$ for zeroth-order Bessel $J_0(\cdot)$ and carrier wavelength $\lambda$~\cite[cf.~Eq.~(3.26)]{Goldsmith2005}.
The recursion is first-order (Jakes at lag one) and preserves $\mathbb E\{\|\vec{g}_u[t]\|^2\} = N_{\rm t}$ since $|J_0(x)|\leq1$.
Because $\psi_u[t]$ tracks the range, a tangentially moving node still decorrelates through \eqref{eq:gauss_markov} with its specular phase nearly fixed while its LoS angle changes fastest. Equations \eqref{eq:rician_channel}--\eqref{eq:gauss_markov} form a stochastic generator requiring no prerecorded data.

\subsection{Monostatic Target, Clutter, and Doppler Model}
\label{subsec:sensing_channel}

With $r_{\rm T}[t] \equiv d_{\rm T}[t]$, the round-trip delay and Doppler frequency are $\tau_{\rm T}[t] = 2r_{\rm T}[t]/c$ and $f_{D,{\rm T}}[t] = -2v_{{\rm r},{\rm T}}[t]/\lambda$ for speed of light $c$.
With $t_k =k T_{\rm s}$ the intra-block sample time, the round-trip phase is $-2\pi f_c\tau_{\rm T}[t] + 2\pi f_{D,{\rm T}}[t] t_k = -4\pi (r_{\rm T}[t] + v_{{\rm r},{\rm T}}[t] t_k)/\lambda$, which is the instantaneous two-way range phase.
The clutter Doppler frequencies $f_{D,{\rm cl},\ell}[t]$ follow the same convention. The narrowband rank-one target channel is
\begin{equation}
    \mathbf H_{\rm T}[t]
    =\alpha_{\rm T}[t]
    \vec{a}_{\rm r}(\phi_{\rm T}[t])
    \vec{a}_{\rm t}^{\mathsf H}(\phi_{\rm T}[t]),
    \label{eq:target_channel}
\end{equation}
where
\begin{equation}
    \alpha_{\rm T}[t]
    =
    \sqrt{G_{\rm t}G_{\rm r}\lambda^2\sigma_{\rm RCS}
    \big/\big[(4\pi)^3r_{\rm T}^4[t]\big]}\,
    e^{-\mathsf j2\pi f_c\tau_{\rm T}[t]}.
    \label{eq:target_gain}
\end{equation}
with per-element gains $G_{\rm t}, G_{\rm r}$, RCS $\sigma_{\rm RCS}$, lognormal across tracks and constant within a track, and carrier frequency $f_c$~\cite[cf.~Eq.~(2.8)]{Richards2014}.
Array gains are already carried by $\|\vec{a}_{\rm t}\|^2 = N_{\rm t}$ and $\|\vec{a}_{\rm r}\|^2 = N_{\rm r}$ in \eqref{eq:target_channel}, so element gains in \eqref{eq:target_gain} avoid double-counting the beamforming gain.
Multipath is represented by $L_{\rm cl}$ scatterers with $\mathbf H_{{\rm cl},\ell}[t] = \alpha_{{\rm cl},\ell}[t] \vec{a}_{\rm r} (\phi_{{\rm cl},\ell}[t]) \vec{a}_{\rm t}^{\mathsf H} (\phi_{{\rm cl},\ell}[t])$, $\alpha_{{\rm cl},\ell} \sim \mathcal{CN}(0, \sigma_{{\rm cl},\ell}^2)$ or from \eqref{eq:target_gain} at a random scatterer range and RCS.
Subscript ${\rm cl}$ denotes clutter, superscript ${\rm c}$ communication.

\subsection{Received Communication Signal and User SINR}

The received signal at communication user $u$ is $y_u[k,t] = \vec{h}_u^{\mathsf H}[t] \vec{v}_u s_u[k] + \sum_{u'\ne u} \vec{h}_u^{\mathsf H}[t] \vec{v}_{u'} s_{u'}[k] + \vec{h}_u^{\mathsf H}[t] \vec{v}_{\rm s} s_{\rm s}[k] + n_u[k]$, with $n_u[k] \sim \mathcal{CN}(0, \sigma_u^2)$.
Although $s_{\rm s}[k]$ is known at the BS, we conservatively assume it is unknown at the users, so each treats every other stream as interference, giving SINR
\begin{equation}
    \gamma_u^{\rm c}[t]
    \!=\! |\vec{h}_u^{\mathsf H}[t] \vec{v}_u|^2 \!\Big/\! \Big( \!
    \textstyle \sum_{u'\ne u}\! |\vec{h}_u^{\mathsf H}[t] \vec{v}_{u'}\!|^2
    \!+ |\vec{h}_u^{\mathsf H}[t] \vec{v}_{\rm s}|^2
    \!+ \sigma_u^2 \Big).
    \label{eq:communication_sinr}
\end{equation}
Spectral efficiencies (SE) are $R_u[t] = \log_2(1 + \gamma_u^{\rm c}[t])$ and $\mathcal C[t] = \sum_{u=1}^{U} R_u[t]$ in bit/s/Hz, while absolute rates are $B_WR_u[t]$ and $B_W\mathcal C[t]$ for bandwidth $B_W$. Working in bit/s/Hz keeps the rate constraints independent of $B_W$.

\subsection{Received Radar Echo and Sensing SINR}
\label{subsec:sensing_signal}

The BS receives target echo, clutter, self-interference, and noise. With sensing sample period $T_{\rm s}$ and number of samples $N_{\rm s}$ in the CPI, $t_k=kT_{\rm s}$ for $k=0,\ldots,N_{\rm s}-1$ referenced to the start of block $t$ where $r_{\rm T}[t],\tau_{\rm T}[t]$ are evaluated, and $N_{\rm s}T_{\rm s}\leq T_{\rm b}$ so the CPI lies inside one block.
The received sensing vector is modeled as
\begin{equation}
\begin{split}
    \vec{y}_{\rm s}[k,t]
    =~& e^{\mathsf j 2\pi f_{D,{\rm T}}[t] t_k}
    \mathbf H_{\rm T}[t] \vec{x}[k-\bar\tau_{\rm T}[t]]\\
    &+ \textstyle \sum_{\ell=1}^{L_{\rm cl}}
    e^{\mathsf j 2\pi f_{D,{\rm cl},\ell}[t] t_k}
    \mathbf H_{{\rm cl},\ell}[t] \vec{x}[k-\bar\tau_\ell[t]]\\
    &+ \sqrt{\chi_{\rm SI}}\, \mathbf H_{\rm SI}[t] \vec{x}[k]
    + \vec{n}_{\rm r}[k],
    \label{eq:received_sensing}
\end{split}
\end{equation}
with delays in samples $\bar\tau_{\rm T}[t]=\operatorname{round}(\tau_{\rm T}[t]/T_{\rm s})$ and $\bar\tau_\ell[t]$.
$\mathbf H_{\rm SI} \in \mathbb C^{N_{\rm r} \times N_{\rm t}}$ is the transmit-to-receive SI channel before cancellation. It is a property of the array geometry and the radio hardware rather than of the propagation environment, so it is drawn once per trajectory and held fixed, with i.i.d. $\mathcal{CN} (0, 1/N_{\rm t})$ entries. The variance gives $\mathbb E\{\|\mathbf H_{\rm SI}\|_{\mathrm F}^2\}=N_{\rm r}$, one unit of coupled power per receive antenna independent of $N_{\rm t}$, so the SI level does not drift with array size.
$0\leq \chi_{\rm SI} \leq1$ is the residual SI power level after cancellation and $\vec{n}_{\rm r}[k] \sim \mathcal{CN}(\vec{0}, \sigma_{\rm r}^2 \mathbf I_{N_{\rm r}})$.
The Doppler phases sit in \eqref{eq:received_sensing} rather than \eqref{eq:target_channel}, so they vary with $k$ and are resolvable by matched filtering at the target delay, where each has unit modulus and is common to every stream of its scatterer, altering neither \eqref{eq:sensing_sinr} nor \eqref{eq:interference_covariance}.
A unit-norm combiner $\vec{u}_{\rm s}\in\mathbb C^{N_{\rm r}}$ gives $z_{\rm s}[k,t]=\vec{u}_{\rm s}^{\mathsf H}\vec{y}_{\rm s}[k,t]$.

The receiver matched-filters $z_{\rm s}$ against $s_{\rm s}$ over the CPI and reads the target range bin, though Doppler is unresolved at these CPIs.
Coherent integration of $N_{\rm s}$ samples gives processing gain $G_{\rm p}=N_{\rm s}$ on $s_{\rm s}$-driven returns in the target echo bin and none on returns uncorrelated with $s_{\rm s}[k]$, which are the target-reflected data streams, noise, and delay-mismatched returns such as clutter and residual SI.
All are therefore sensing interference.
The sensing SINR is
\begin{equation}
\begin{split}
    \gamma^{\rm s}[t]
    ={}&G_{\rm p}|\vec{u}_{\rm s}^{\mathsf H}\mathbf H_{\rm T}[t]
    \vec{v}_{\rm s}|^2\Big/\Big(
    \textstyle\sum_{u=1}^{U}|\vec{u}_{\rm s}^{\mathsf H}
    \mathbf H_{\rm T}[t]\vec{v}_u|^2\\
    &+\vec{u}_{\rm s}^{\mathsf H}\mathbf R_{\rm in}[t]
    \vec{u}_{\rm s}\Big),
    \label{eq:sensing_sinr}
\end{split}
\end{equation}
where $\mathbf R_{\rm in}[t]$ is the interference covariance after matched-filtering, given by
\begin{equation}
\begin{split}
    \mathbf R_{\rm in}[t]
    ={}&\sum_{\ell=1}^{L_{\rm cl}}
    \mathbf H_{{\rm cl},\ell}[t]
    \mathbf R_x
    \mathbf H_{{\rm cl},\ell}^{\mathsf H}[t]\\
    &+\chi_{\rm SI}\mathbf H_{\rm SI}[t]
    \mathbf R_x\mathbf H_{\rm SI}^{\mathsf H}[t]
    +\sigma_{\rm r}^2\mathbf I_{N_{\rm r}}.
    \label{eq:interference_covariance}
\end{split}
\end{equation}
Since scatterer distances differ from $r_{\rm T}$, the delays $\bar\tau_\ell$ miss the correlator at $\bar\tau_{\rm T}$ and the co-range remainder is treated as delay-mismatched too. This is a modeling assumption: a scatterer inside the $c/(2B_W)=1.5$~m range cell of the target would be coherent with $s_{\rm s}$ and receive the full gain $G_{\rm p}$, so for such co-range clutter the model is optimistic.
Clutter therefore enters with the full $\mathbf R_x$ and no processing gain, and the residual SI with the factor $\chi_{\rm SI}$. At these operating ranges, $\gamma^{\rm s}\ge\Gamma_{\rm s}^{\min}$ is attainable only after coherent integration, i.e., with $G_{\rm p}=N_{\rm s}$.
Equation~\eqref{eq:sensing_sinr} assumes a dedicated sensing-stream receiver. Communication symbols are not exploited at the radar processor, so target-reflected communication returns remain interference.

\subsection{Channel Aging}
\label{subsec:csi}

The SINRs \eqref{eq:communication_sinr} and \eqref{eq:sensing_sinr} are realized on the true channels \eqref{eq:rician_channel} and \eqref{eq:target_channel}, whereas the BS computes the precoders and the combiner from its own knowledge of these channels. CSI estimation error is not considered, so the CSI imperfection studied here is channel aging alone.

Channel aging degrades achievable rates in mobile ISAC and is modeled explicitly \cite{Menanor_TVT_2026}. Outdated CSI is the $n$-step evaluation of \eqref{eq:gauss_markov}, designing on $\{\vec{h}_u[t],\mathbf H_{\rm T}[t]\}$ and scoring on $\{\vec{h}_u[t+n],\mathbf H_{\rm T}[t+n]\}$.

\section{Problem Formulation and Feasibility}
\label{sec:problem}

\subsection{Joint Beamforming and Power Allocation Problem}

The communication utility is denoted by $\mathcal C$, while the sensing utility is the dimensionless score $\mathcal S=\log_2(1+\gamma^{\rm s})$, not an achievable sensing rate.
To equalize numerical scale, we set $\widetilde{\mathcal C}=\mathcal C/\mathcal C^{\rm ref}$, $\widetilde{\mathcal S}=\mathcal S/\mathcal S^{\rm ref}$ with $\mathcal C^{\rm ref}=\mathbb E\{\mathcal C\}|_{\omega=1}>0$ and $\mathcal S^{\rm ref}=\mathbb E\{\mathcal S\}|_{\omega=0}>0$ the empirical averages of the analytical design at its two extremes, which are evaluated once per deployment and not per realization, so tradeoff weight $\omega$ maintains a consistent interpretation across all channel realizations.
The joint optimization problem is formulated as
\begin{subequations}
\label{eq:isac_problem}
\begin{align}
    \underset{\substack{\vec{p},\,p_{\rm s},\, \mathbf W,\, \vec{w}_{\rm s},\,\vec{u}_{\rm s}}} {\operatorname{maximize}}
    \quad&
    \omega \widetilde{\mathcal C} + (1-\omega) \widetilde{\mathcal S}
    \label{eq:isac_objective}\\
    \operatorname{subject\ to}\quad
    & \gamma^{\rm s} \geq \Gamma_{\rm s}^{\min},
    \label{eq:sensing_constraint}\\
    & R_u\geq R_u^{\min},\quad u=1,\ldots,U,
    \label{eq:user_rate_constraint}\\
    & p_{\rm s} + \sum_{u=1}^{U} p_u \leq P_{\max},
    \label{eq:isac_power_constraint}\\
    & p_{\rm s} \geq 0, \quad p_u \geq 0,\quad u=1,\ldots,U,
    \label{eq:nonnegative_power}\\
    & \|\vec{w}_{\rm s}\|^2=1, \quad \|\vec{w}_u\|^2=1, \quad \forall ~u,
    \label{eq:beam_constraints}\\
    & \|\vec{u}_{\rm s}\|^2=1,
    \label{eq:combiner_constraint}
\end{align}
\end{subequations}
where $\vec{p}=[p_1,\ldots,p_U]^{\mathsf T}$ and $\mathbf W=[\vec{w}_1,\ldots,\vec{w}_U]$. The SE thresholds $R_u^{\min}$ are in bit/s/Hz, and together with the sensing threshold $\Gamma_{\rm s}^{\min}$, they keep either function from being sacrificed at extreme $\omega$. Sweeping $\omega$ traces the supported points of the Pareto frontier.
However, deriving the exact Pareto boundary remains a distinct and complex challenge, that is sensitive to the instantaneous worst-case sensing requirement~\cite{Liu_TWC_2025}.

Problem~\eqref{eq:isac_problem} is nonconvex, as both SINRs are coupled fractional quadratics in $(\vec{p},p_{\rm s},\mathbf W, \vec{w}_{\rm s},\vec{u}_{\rm s})$.
Furthermore, mutual interference degrades both metrics.
Sensing leakage degrades $\gamma_u^{\rm c}$ and data leakage degrades $\gamma^{\rm s}$. In addition, since the analytical design enforces no QoS constraint, \eqref{eq:sensing_constraint}--\eqref{eq:user_rate_constraint} must be strictly satisfied per realization, not on the statistical mean.
Finally, $\mathcal S$ is monotone of $\gamma^{\rm s}$, ensuring numerical balance with the communication utility.

\subsection{Closed-Form Sensing SINR and Infeasibility}
\label{subsec:feasibility}

Problem~\eqref{eq:isac_problem} is infeasible on some realizations. For the rank-one target \eqref{eq:target_channel} every transmit stream arrives at the receive array through one spatial direction, $\vec{u}_{\rm s}^{\mathsf H}\mathbf H_{\rm T}\vec{v}_i =\alpha_{\rm T} (\vec{u}_{\rm s}^{\mathsf H}\vec{a}_{\rm r}) (\vec{a}_{\rm t}^{\mathsf H}\vec{v}_i)$, so \eqref{eq:sensing_sinr} becomes
\begin{equation}
    \gamma^{\rm s}
    =\frac{G_{\rm p}A\rho_{\rm s}}
    {B\rho_{\rm s}+1},
    \qquad
    \rho_{\rm s}
    =\frac{|\vec{u}_{\rm s}^{\mathsf H}\vec{a}_{\rm r}|^2}
    {\vec{u}_{\rm s}^{\mathsf H}\mathbf R_{\rm in}\vec{u}_{\rm s}},
    \label{eq:sensing_sinr_closed}
\end{equation}
where $A=p_{\rm s}|\alpha_{\rm T}|^2 |\vec{a}_{\rm t}^{\mathsf H}\vec{w}_{\rm s}|^2$ and $B=|\alpha_{\rm T}|^2\sum_{u}p_u |\vec{a}_{\rm t}^{\mathsf H}\vec{w}_u|^2$ are independent of $\vec{u}_{\rm s}$.
The generalized Rayleigh quotient $\rho_{\rm s}$
is maximized by $\vec{u}_{\rm s}^{\star}\propto\mathbf R_{\rm in}^{-1}\vec{a}_{\rm r}$~\cite[Lemma~B.10]{Bjornson_2017} at $\rho_{\rm s}^{\star}=\vec{a}_{\rm r}^{\mathsf H} \mathbf R_{\rm in}^{-1}\vec{a}_{\rm r}$.
Hence
\begin{equation}
    \gamma^{\rm s\star}
    =\frac{G_{\rm p}A\rho_{\rm s}^{\star}}{B\rho_{\rm s}^{\star}+1}.
    \label{eq:sensing_sinr_opt}
\end{equation}
Since $\rho_{\rm s}^{\star}$ depends on the powers through $\mathbf R_{\rm in}$, raising $p_u$ increases $B$ and expands $\mathbf R_{\rm in}$, both of which decrease $\gamma^{\rm s\star}$.
Thus, the sensing constraint is the most stringent at large communication power.

\begin{theorem}[Sensing-only power split]
\label{thm:max_sensing}
For fixed beams and combiner, $\gamma^{\rm s}$ in \eqref{eq:sensing_sinr} is maximized by the sensing-only split $p_{\rm s}=P_{\max}$, $p_u=0$. With the minimum variance distortionless response (MVDR) combiner $\vec{u}_{\rm s}\propto(\mathbf R_{\rm in}^{({\rm s})})^{-1} \vec{a}_{\rm r}$ and the target-matched beam $\vec{w}_{\rm s}=\vec{a}_{\rm t}(\phi_{\rm T})/\sqrt{N_{\rm t}}$, the resulting sensing SINR admits the closed form
\begin{equation}
    \gamma^{\rm s,max}
    =G_{\rm p}P_{\max}|\alpha_{\rm T}|^2 N_{\rm t}\,
    \vec{a}_{\rm r}^{\mathsf H}
    (\mathbf R_{\rm in}^{({\rm s})})^{-1}
    \vec{a}_{\rm r},
    \label{eq:gamma_s_max}
\end{equation}
where $\mathbf R_{\rm in}^{({\rm s})}$ is \eqref{eq:interference_covariance} under the sensing-only allocation.
\end{theorem}
\begin{IEEEproof}
The proof is given in Appendix~\ref{app:proof_max_sensing}.
\end{IEEEproof}
Theorem~\ref{thm:max_sensing} states that shifting power from any communication stream to sensing increases $\gamma^{\rm s}$, and that maximizing out the combiner preserves this split while yielding the MVDR receiver, which optimally suppresses interference under a distortionless constraint toward $\phi_{\rm T}$~\cite[Eq.~(88)]{Liu_JPROC_2026}. Since the sensing-only state is achievable, $\gamma^{\rm s,max}\ge\Gamma_{\rm s}^{\min}$ is a sufficient but not necessary condition for \eqref{eq:sensing_constraint}. In practice, any $R_u^{\min}>0$ forbids this sensing-only vertex, and a clutter-aware $\vec{w}_{\rm s}$ still affects $\mathbf R_{\rm in}^{({\rm s})}$ through illumination.

Since $\mathbf R_{\rm in}\succeq\sigma_{\rm r}^2\mathbf I_{N_{\rm r}}$, it follows that $\rho_{\rm s}^{\star}\le N_{\rm r}/\sigma_{\rm r}^2$. Thus,
\begin{equation}
    \gamma^{\rm s,max}
    \le\gamma_{\rm NL}
    \triangleq
    \frac{G_{\rm p}P_{\max}|\alpha_{\rm T}|^2 N_{\rm t}N_{\rm r}}
    {\sigma_{\rm r}^2}.
    \label{eq:noise_limited}
\end{equation}
The term $\gamma_{\rm NL}$ is the noise-limited monostatic SINR after coherent integration and upper-bounds all achievable $\gamma^{\rm s}$. Substituting \eqref{eq:target_gain} into this bound gives a necessary geometric condition for feasibility.

\begin{corollary}[Necessary range--RCS--power condition]
\label{cor:range}
If
\begin{equation}
    r_{\rm T}^{4}
    >
    \frac{G_{\rm p}P_{\max}G_{\rm t}G_{\rm r}\lambda^{2}
    \sigma_{\rm RCS}N_{\rm t}N_{\rm r}}
    {(4\pi)^{3}\sigma_{\rm r}^{2}\Gamma_{\rm s}^{\min}},
    \label{eq:range_condition}
\end{equation}
then \eqref{eq:sensing_constraint} remains infeasible regardless of the power allocation, transmit beams, and receive combiner.
Equivalently, Problem~\eqref{eq:isac_problem} has no feasible solution on any realization where $\gamma_{\rm NL}<\Gamma_{\rm s}^{\min}$.
\end{corollary}

This noise-limited bound is necessary but not sufficient, as clutter and residual SI inflate $\mathbf R_{\rm in}$ above $\sigma_{\rm r}^2\mathbf I$.
Because the interference rank at the sensing-only allocation is at most $L_{\rm cl}+1$, a noise-dominated subspace survives whenever this rank is below $N_{\rm r}$ and $\vec{a}_{\rm r}(\phi_{\rm T})$ extends outside the interference span, allowing $\gamma^{\rm s}$ to grow without bound in $p_{\rm s}$.
Conversely, if a clutter patch perfectly aligns with the target angle, $\gamma^{\rm s}$ saturates at a signal-to-clutter-and-noise ratio (SCNR) ceiling. Resolving this alignment conventionally requires Doppler space--time processing~\cite[Secs.~V-B, V-C]{Liu_JPROC_2026}, which is unavailable in our receiver (Section~\ref{subsec:sensing_signal}).
Finally, while joint transmission raises the SI rank to $U+1$ and potentially spans $\mathbb C^{N_{\rm r}}$ alongside clutter, the operating $\chi_{\rm SI}=-110$~dB suppression pushes this SCNR ceiling far above $P_{\max}$. Therefore, at finite $P_{\max}$, the power budget is the true binding limit, except during exact target--clutter alignment.

We define the analytic per-realization sensing floor evaluated at the allocation in \eqref{eq:gamma_s_max} as
\begin{equation}
    \check g_{\rm s}
    =\frac{[\Gamma_{\rm s}^{\min}-\gamma^{\rm s,max}]_+}
    {\Gamma_{\rm s}^{\min}+\epsilon_0},
    \label{eq:g_floor_analytic}
\end{equation}
 where $[x]_+=\max\{0,x\}$.
 This measures the normalized violation at the sensing-only matched-beam point, which falls to zero when the threshold $\Gamma_{\rm s}^{\min}$ is met.
 Its expectation $\check{\varepsilon}_{\rm s}=\mathbb E\{\check g_{\rm s}\}$ is $\omega$-independent, since $\gamma^{\rm s,max}$ depends solely on the realization and $P_{\max}$.

It is important to distinguish this exhibited floor from an absolute physical limit.
Corollary~\ref{cor:range} rigorously establishes that the violation cannot vanish globally, since $\gamma_{\rm NL}$ bounds $\gamma^{\rm s}$ for every beam, split, and combiner. Thus, the noise-limited floor $\varepsilon_{\rm NL}=\mathbb E\{[\Gamma_{\rm s}^{\min}-\gamma_{\rm NL}]_+/ (\Gamma_{\rm s}^{\min}+\epsilon_0)\}$ lower-bounds the expected normalized sensing violation of every design, and $\varepsilon_{\rm NL}>0$ whenever \eqref{eq:range_condition} holds on a set of positive measure.
Conversely, $\check{\varepsilon}_{\rm s}$ only measures the violation at one specific allocation. Because Theorem~\ref{thm:max_sensing} is sufficient but not necessary, a clutter-aware $\vec{w}_{\rm s}$ might theoretically achieve a lower violation. Therefore, the minimal attainable violation $\varepsilon_{\rm s}^{\star}$ satisfies $\varepsilon_{\rm NL}\le\varepsilon_{\rm s}^{\star}\le\check{\varepsilon}_{\rm s}$.

\subsection{Sensing Floor at the Communication-Heavy Extreme}

Any $R_u^{\min}>0$ requires $p_u>0$, so the sensing-only allocation of Theorem~\ref{thm:max_sensing} is never jointly feasible, and a split $p_{\rm s}/P_{\max}\in(0,1)$ is needed to serve both functions. With $p_u>0$, the leakage $B$ in \eqref{eq:sensing_sinr_closed} is positive and $\mathbf R_{\rm in}$ contains the clutter and SI driven by the user streams. By Step~2 of Appendix~\ref{app:proof_max_sensing}, the sensing violation of a jointly feasible design is therefore at least that of the sensing-only split with the same beams and combiner, and strictly larger wherever $\gamma^{\rm s}<\Gamma_{\rm s}^{\min}$.
This excess grows with $\omega$, since the objective \eqref{eq:isac_objective} then rewards moving power from sensing to the users, while \eqref{eq:sensing_constraint} holds $p_{\rm s}/P_{\max}$ above a realization-dependent threshold. A design that keeps more sensing power than the objective alone would choose thus gives up mean utility for reliability, whereas the analytical prior, which sets the split from $\omega$ alone, makes the opposite trade.

\section{Physics-Informed Unsupervised Learning}
\label{sec:unsl_solution}

PI-UnSL maps the communication, target, clutter, and SI channels directly to an allocation, and trains on the objective and constraints of \eqref{eq:isac_problem} without optimal labels~\cite{Triwidyastuti_TVT_2026}.
Fig.~\ref{fig:signal_chain} places it in the signal chain of Section~\ref{sec:system_model}. From the CSI and the priority $\omega$ (or tradeoff weight), one forward pass sets the power split and transmit beams at the BS transmitter and the combiner at its sensing receiver, while the users and the propagation are left unchanged. During training, the same SINR model that scores the link supplies the label-free loss.

\begin{figure*}[!t]
    \centering
    \resizebox{\textwidth}{!}{\definecolor{tndblue}{HTML}{0018A8}%
\definecolor{tndpurple}{HTML}{6B3FA0}%
\definecolor{tndgreen}{HTML}{086C46}%
\definecolor{tndgray}{HTML}{6E6E6E}%
\definecolor{tndred}{HTML}{ED1B24}%
\definecolor{tndcyan}{HTML}{00A9CE}%
\definecolor{tndgold}{HTML}{D4AA00}%
\definecolor{tndorange}{HTML}{FF9800}%
\begin{tikzpicture}[
    font=\scriptsize,
    >={Latex[length=1.5mm,width=1.3mm]},
    box/.style={draw, rounded corners=1.5pt, align=center, minimum height=0.9cm,
                inner sep=2.2pt, line width=0.6pt},
    sys/.style={box, draw=tndgray},
    intf/.style={box, draw=tndred},
    txc/.style={box, draw=tndblue, line width=1pt},
    rxc/.style={box, draw=tndpurple, line width=1pt},
    perf/.style={box, draw=tndorange!80!black},
    inp/.style={box, draw=tndcyan!75!black},
    prior/.style={box, draw=tndgold!75!black},
    prop/.style={box, draw=tndgreen, line width=1pt},
    propt/.style={prop, dashed},
    sumn/.style={draw, circle, inner sep=0pt, minimum size=4.2mm, line width=0.6pt},
    sig/.style={->, line width=0.6pt},
    ctl/.style={->, line width=1.1pt, draw=tndgreen},
    csi/.style={->, line width=0.8pt, draw=tndcyan!75!black, dash pattern=on 3pt off 1.5pt},
    trn/.style={->, line width=0.8pt, draw=tndgreen, dash pattern=on 3pt off 1.5pt},
    mdl/.style={->, line width=0.6pt, draw=tndorange!80!black, densely dotted},
    halo/.style={preaction={draw, white, line width=3pt, -}},
    band/.style={draw=#1, rounded corners=4pt, line width=0.5pt, dash pattern=on 2pt off 1.5pt},
    title/.style={font=\scriptsize\bfseries, text=#1, anchor=south west, inner sep=1pt},
    dot/.style={circle, fill, inner sep=0pt, minimum size=1.6pt},
]

\draw[band=tndgray] (-0.1,-0.75) rectangle (19.25,4.3);
\node[title=tndgray] at (-0.05,4.32)
    {Mobile multiuser MIMO-ISAC link, block $t$ (BS with $N_{\rm t}$ transmit and $N_{\rm r}$ receive antennas)};

\node[sys, minimum width=1.7cm] (src) at (0.9,2.0)
    {Unit-energy\\symbols\\$s_{\rm s}[k]$, $s_u[k]$};
\node[txc, minimum width=2.0cm] (pow) at (3.2,2.0)
    {Power split\\$p_{\rm s},p_1,\ldots,p_U$\\$\textstyle\sum_i p_i=P_{\max}$};
\node[txc, minimum width=2.0cm] (bf) at (5.6,2.0)
    {Transmit beams\\$\vec{w}_{\rm s},\vec{w}_u$\\$\vec{v}_i=\sqrt{p_i}\,\vec{w}_i$};
\node[sumn] (sp) at (7.3,2.0) {$+$};
\node[font=\tiny, align=center, below=1pt] at (sp.south) {super-\\position};
\draw[sig] (src) -- (pow);
\draw[sig] (pow) -- (bf);
\draw[sig] (bf) -- (sp);

\node[sys,  minimum width=2.6cm] (hc)  at (9.9,3.45)
    {Rician mobile channel\\$\vec{h}_u[t]$, $u=1,\ldots,U$};
\node[sys,  minimum width=2.6cm] (ht)  at (9.9,2.0)
    {Target echo $\mathbf H_{\rm T}[t]$\\delay $\bar\tau_{\rm T}$, Doppler\\\textcolor{tndred}{+ data echoes $\mathbf H_{\rm T}\vec{v}_u$}};
\node[intf, minimum width=2.6cm, minimum height=0.62cm] (hcl) at (9.9,0.75)
    {Clutter $\mathbf H_{{\rm cl},\ell}[t]$, $\ell\le L_{\rm cl}$};
\node[intf, minimum width=2.6cm, minimum height=0.62cm] (hsi) at (9.9,-0.2)
    {Residual SI $\sqrt{\chi_{\rm SI}}\,\mathbf H_{\rm SI}$};
\draw[line width=0.6pt] (sp) -- node[above]{$\vec{x}[k]$} (8.1,2.0);
\draw[line width=0.6pt] (8.1,3.45) -- (8.1,-0.2);
\foreach \n in {hc,ht,hcl,hsi} {
    \draw[sig] (\n.west -| 8.1,0) -- (\n.west);
}
\node[dot] at (8.1,2.0) {};

\node[sys, minimum width=2.6cm] (ue) at (13.8,3.45)
    {User $u$ receiver\\$y_u=\vec{h}_u^{\mathsf H}\vec{x}[k]+n_u$\\\textcolor{tndred}{$\vec{v}_{\rm s}$, $\vec{v}_{u'\ne u}$ interfere}};
\node[perf, minimum width=2.3cm] (rate) at (17.9,3.45)
    {SINR $\gamma_u^{\rm c}$, SE $R_u$\\QoS: $R_u\ge R_u^{\min}$};
\draw[sig] (hc) -- (ue);
\draw[sig] (ue) -- (rate);

\node[sumn] (rs) at (11.7,0.75) {$+$};
\node[font=\tiny, anchor=north west, inner sep=0.5pt] at (rs.south east) {$+\vec{n}_{\rm r}$};
\node[rxc, minimum width=1.9cm] (comb) at (13.4,0.75)
    {Receive\\combiner\\$z_{\rm s}=\vec{u}_{\rm s}^{\mathsf H}\vec{y}_{\rm s}$};
\node[sys, minimum width=1.8cm] (mf) at (15.55,0.75)
    {Matched filter,\\coherent gain\\$G_{\rm p}=N_{\rm s}$};
\node[perf, minimum width=2.3cm] (sens) at (17.9,0.75)
    {Sensing SINR $\gamma^{\rm s}$, $P_{\rm D}$\\QoS: $\gamma^{\rm s}\ge\Gamma_{\rm s}^{\min}$};
\draw[sig] (ht.east) -| (rs.north);
\draw[sig] (hcl.east) -- (rs.west);
\draw[sig] (hsi.east) -| (rs.south);
\draw[sig] (rs) -- node[above]{$\vec{y}_{\rm s}$} (comb);
\draw[sig] (comb) -- (mf);
\draw[sig] (mf) -- (sens);

\draw[band=tndgreen] (-0.1,-1.25) rectangle (19.25,-5.0);
\node[title=tndgreen, anchor=north west] at (-0.05,-5.02)
    {Proposed PI-UnSL: one forward pass per block at inference; dashed blocks and arrows are used in training only};

\node[inp, minimum width=2.0cm] (in) at (1.1,-2.45)
    {Priority $\omega$\\QoS $\Gamma_{\rm s}^{\min}$, $R_u^{\min}$\\budget $P_{\max}$};
\node[prior, minimum width=2.4cm] (pri) at (3.8,-2.45)
    {Closed-form baseline\\RZF $\vec{w}_u^{(0)}$, matched $\vec{w}_{\rm s}^{(0)}$\\MVDR $\vec{u}_{\rm s}^{(0)}$, powers $f_i^{(0)}$};
\node[prop, minimum width=2.4cm] (qs) at (6.65,-2.45)
    {Physics statistics $\vec{q}_{\rm sys}$\\gains, leakage,\\feasibility margin};
\node[prop, minimum width=2.8cm] (nn) at (9.9,-2.45)
    {Residual network\\encoder $E$, FiLM on $\vec{q}_{\rm sys}$\\heads $D_{\rm p}$, $D_{\rm t}$, $D_{\rm r}$};
\node[sumn, draw=tndgreen, line width=1pt] (res) at (11.95,-2.45) {$+$};
\node[prop, minimum width=2.5cm] (proj) at (14.2,-2.45)
    {Exact normalizations\\power simplex,\\unit-norm $\vec{w}_i$, $\vec{u}_{\rm s}$};

\draw[sig] (in) -- (pri);
\draw[sig] (pri) -- (qs);
\draw[sig] (qs) -- (nn);
\draw[sig] (nn) -- node[above]{$\widetilde{\boldsymbol\Theta}$} (res);
\draw[sig] (res) -- (proj);
\draw[sig] ([xshift=4mm]pri.south) |- (11.95,-3.2) node[pos=0.7, below]{$\boldsymbol\Theta^{(0)}$} -- (res.south);

\node[propt, minimum width=2.5cm] (floor) at (6.1,-4.3)
    {Measured violation\\floor $\check g_{\rm s}(\omega_k)$};
\node[propt, minimum width=3.4cm] (loss) at (9.9,-4.3)
    {Calibrated PHR loss: per-bin\\duals $\mu_i(\omega_k)$, user CVaR,\\achievability gap $\tilde g_{\rm s}$};
\node[propt, minimum width=2.6cm] (dm) at (17.9,-3.4)
    {Differentiable SINRs, SE\\$\gamma_u^{\rm c}$, $R_u$, $\gamma^{\rm s}$};
\draw[trn] ([xshift=-4mm]pri.south) |- (floor.west);
\draw[trn] (floor) -- (loss);
\draw[trn] (proj.east) -| ([xshift=-5mm]dm.north);
\draw[trn] (dm.south) |- (loss.east);
\draw[trn, halo] (loss.north) -- node[right, pos=0.35]{$\nabla_{\boldsymbol\Gamma}$} (nn.south);

\draw[line width=1.1pt, draw=tndgreen] (proj.north) -- (14.2,-0.95) -- (3.2,-0.95);
\draw[ctl] (3.2,-0.95) -- node[right, pos=0.2, text=tndgreen]{$p_{\rm s}, p_u$} (pow.south);
\draw[ctl] (5.6,-0.95) -- node[right, pos=0.2, text=tndgreen]{$\vec{w}_{\rm s}, \vec{w}_u$} (bf.south);
\draw[ctl] (13.4,-0.95) -- node[right, pos=0.3, text=tndgreen]{$\vec{u}_{\rm s}$} (comb.south);
\node[dot, tndgreen, minimum size=2.2pt] at (5.6,-0.95) {};
\node[dot, tndgreen, minimum size=2.2pt] at (13.4,-0.95) {};

\draw[csi, halo] (hsi.south) -- (nn.north);
\draw[csi] (9.9,-1.7) -| (pri.north);
\node[dot, tndcyan!75!black] at (9.9,-1.7) {};
\node[font=\tiny, text=tndcyan!60!black, anchor=south] at (5.2,-1.7)
    {CSI $\mathcal X_{\rm in}$ from block $t-n$};

\draw[mdl] ([xshift=6mm]sens.south) -- node[left, pos=0.72, font=\tiny, text=tndorange!70!black, align=right]{same\\model} ([xshift=6mm]dm.north);

\begin{scope}[shift={(0.0,-5.75)}, every node/.style={font=\scriptsize}]
    \draw[draw=tndgreen, line width=1pt] (0.0,0) rectangle ++(0.32,0.2);
    \node[anchor=west] at (0.38,0.1) {Proposed};
    \draw[draw=tndblue, line width=1pt] (1.75,0) rectangle ++(0.32,0.2);
    \node[anchor=west] at (2.13,0.1) {Tx blocks set by PI-UnSL};
    \draw[draw=tndpurple, line width=1pt] (5.35,0) rectangle ++(0.32,0.2);
    \node[anchor=west] at (5.73,0.1) {Rx block set by PI-UnSL};
    \draw[draw=tndgold!75!black] (8.85,0) rectangle ++(0.32,0.2);
    \node[anchor=west] at (9.23,0.1) {Baseline};
    \draw[draw=tndorange!80!black] (11.45,0) rectangle ++(0.32,0.2);
    \node[anchor=west] at (11.83,0.1) {QoS metrics};
    \draw[draw=tndred] (13.55,0) rectangle ++(0.32,0.2);
    \node[anchor=west] at (13.93,0.1) {Impairments};
    \draw[draw=tndgray] (15.75,0) rectangle ++(0.32,0.2);
    \node[anchor=west] at (16.13,0.1) {Fixed physics};
\end{scope}
\end{tikzpicture}%
}
    \caption{Signal chain of the MIMO-ISAC link with PI-UnSL (green), which sets the transmit blocks (blue) and the radar combiner (purple) from possibly $n$-block-old CSI (cyan). Dashed green paths are used in training only.}
    \label{fig:signal_chain}
\end{figure*}
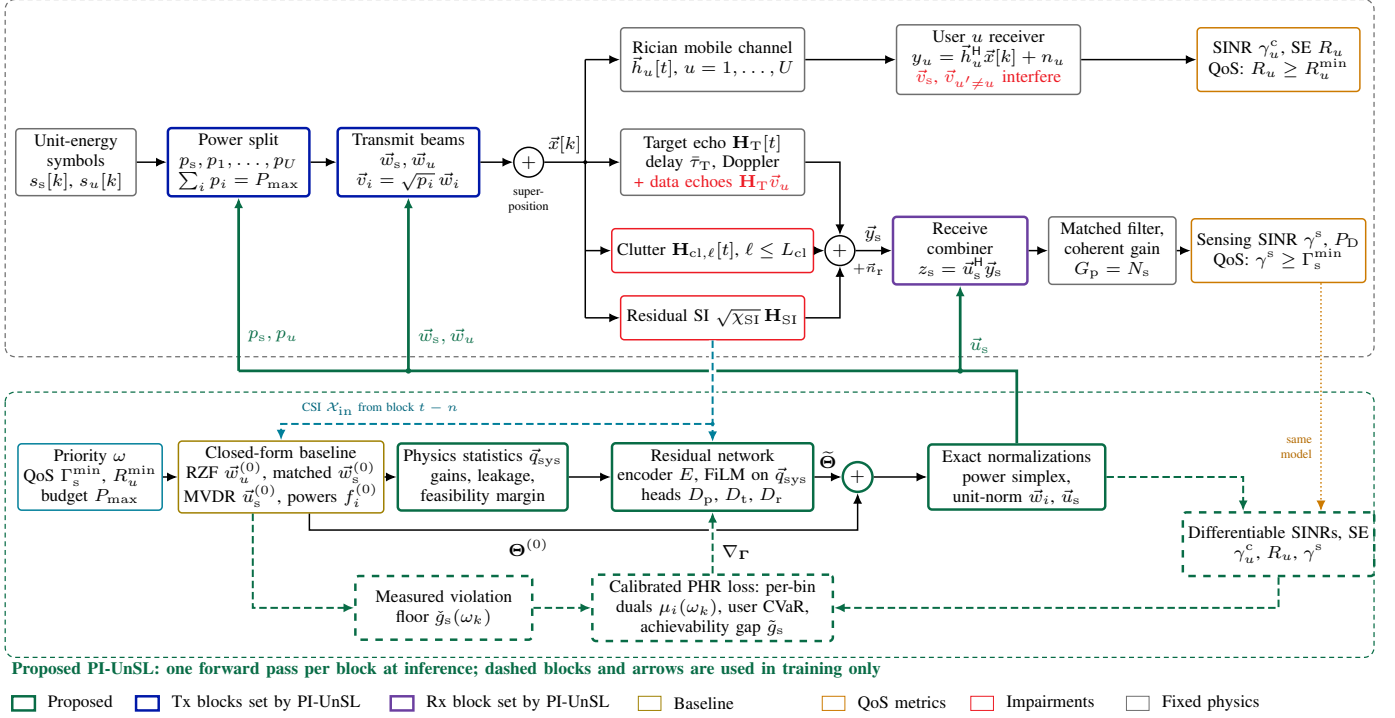

\subsection{Channel Input and Physics-Informed Conditioning}

With $\mathbf H_{\rm c}=[\vec{h}_1,\ldots,\vec{h}_U]^{\mathsf H} \in \mathbb C^{U\times N_{\rm t}}$, the channel-state matrix is $\mathbf H_{\rm in}=[\mathbf H_{\rm c}^{\mathsf T},\mathbf H_{\rm T}^{\mathsf T}, \mathbf H_{{\rm cl},1}^{\mathsf T},\ldots,\mathbf H_{{\rm cl},L_{\rm cl}}^{\mathsf T}, \mathbf H_{\rm SI}^{\mathsf T}]^{\mathsf T}\in\mathbb C^{N_{\rm in}\times N_{\rm t}}$, $N_{\rm in}=U+(L_{\rm cl}+2)N_{\rm r}$, zero-padded when fewer clutter components are present. The input tensor $\mathcal X_{\rm in} \in \mathbb R^{2\times N_{\rm in}\times N_{\rm t}}$ holds $\mathcal X_{\rm in}(1,:,:)=\Re\{\mathbf H_{\rm in}\}$, $\mathcal X_{\rm in}(2,:,:)=\Im\{\mathbf H_{\rm in}\}$, and the inserted physics information is
\begin{equation}
\begin{split}
	\vec{q}_{\rm sys} =& \big[ \omega, P_{\max},\Gamma_{\rm s}^{\min}, R_1^{\min},\ldots,R_U^{\min}, \sigma_{\rm r}^2,\sigma_1^2,\ldots,\sigma_U^2, \\ & \mathbf{G}_{\rm comm}, g_{\rm sens}, \vec{l}_{\rm sens}, \mu_{\rm feas} \big]^{\mathsf T}
	\label{eq:qsys}
\end{split}
\end{equation}
that conditions the latent.
The interface takes $P_{\max},\Gamma_{\rm s}^{\min}$ and the noise powers logarithmically and the rate thresholds normalized by $\mathcal C^{\rm ref}$.

The remaining terms are the closed-form sufficient statistics of the interference problem. $\mathbf{G}_{\rm comm} \in \mathbb{R}^{U \times (U+1)}$ is the matrix of noise-normalized communication gains, where $ [\mathbf G_{\rm comm}]_{u,j} = 10 \log_{10} (| \vec{h}_u^{\mathsf H} \vec{w}_j^{(0)} |^2 / \sigma_u^2) $.
Target sensing gain is $g_{\rm sens}  = 10 \log_{10} \big( |(\vec{u}_{\rm s}^{(0)})^{\mathsf H} \mathbf H_{\rm T} \vec{w}_{\rm s}^{(0)}|^2 / (\vec{u}_{\rm s}^{(0)})^{\mathsf H}\mathbf R_{\rm in}^{(0)}\vec{u}_{\rm s}^{(0)} \big)$.
$\vec{l}_{\rm sens}$ is the $U$-dimensional vector of sensing leakage terms that can be expressed as $ [\vec{l}_{\rm sens}]_u = 10 \log_{10} \big( |(\vec{u}_{\rm s}^{(0)})^{\mathsf H} \mathbf H_{\rm T} \vec{w}_u^{(0)}|^2 / (\vec{u}_{\rm s}^{(0)})^{\mathsf H}\mathbf R_{\rm in}^{(0)}\vec{u}_{\rm s}^{(0)} \big) $, which represents how much interference each communication beam leaks into the radar receiver.
$\mu_{\rm feas} = 10 \log_{10} ( G_{\rm p} P_{\max} |\alpha_{\rm T}|^2 |\vec{a}_{\rm t}^{\mathsf H} \vec{w}_{\rm s}^{(0)}|^2 \cdot \vec{a}_{\rm r}^{\mathsf H} (\mathbf R_{\rm in}^{(0)})^{-1} \vec{a}_{\rm r} / \Gamma_{\rm s}^{\min} )$ is the analytical feasibility margin. All statistics are evaluated at the closed-form baseline of Section~\ref{subsec:model_prior}, with $\mathbf R_{\rm in}^{(0)}$ its interference covariance and $\vec{u}_{\rm s}^{(0)}$ its combiner. By explicitly providing these exact physical statistics, the network bypasses the need to re-learn the interference polynomials from the raw channel matrices.
Since not every entry of \eqref{eq:qsys} is equally informative, a fixed element-wise weight $\vec{w}_{\rm sys} \in [0,1]^{|\vec{q}_{\rm sys}|}$ is applied, where $\vec{q}_{\rm sys} \leftarrow \vec{w}_{\rm sys} \odot \vec{q}_{\rm sys}$. Since $P_{\max}$, $\Gamma_{\rm s}^{\min}$, $R_u^{\min}$, and the noise powers are fixed in this study, they receive weight $0$. The weight is $1$ for $\omega$, the sensing-stream and own-beam gains in $\mathbf G_{\rm comm}$, and $\mu_{\rm feas}$, and $0.5$ for the cross-user gains, $g_{\rm sens}$, and $\vec{l}_{\rm sens}$. These values were set from the group-lasso importance of a preliminary training run.

Let $E_{\boldsymbol{\Gamma}_{\rm e}}(\cdot)$ be a shared convolutional encoder. Each row of $\mathcal X_{\rm in}$ is scale-normalized and the discarded scale re-injected as $\log_{10}$ in a third map, preserving path-loss-invariant structure alongside absolute magnitude. A feature-wise linear modulation (FiLM) layer $M_{\boldsymbol{\Gamma}_{\rm m}}$ conditions the latent and three heads produce
\begin{equation}
\begin{split}
    \vec{z} = E_{\boldsymbol{\Gamma}_{\rm e}} (\mathcal X_{\rm in}), \quad& (\vec{\alpha}_{\rm m}, \vec{\beta}_{\rm m}) = M_{\boldsymbol {\Gamma}_{\rm m}} (\vec{q}_{\rm sys}),\\
    \widetilde{\vec{z}} = (\vec{1} + \vec{\alpha}_{\rm m}) & \odot \vec{z} + \vec{\beta}_{\rm m},\\
    (\widetilde {\vec{a}}, \widetilde {\mathbf W}_{\rm all}, \widetilde {\vec{u}}_{\rm s}) = \bigl( D_{\boldsymbol {\Gamma}_{\rm p}},& D_{\boldsymbol {\Gamma}_{\rm t}}, D_{\boldsymbol {\Gamma}_{\rm r}} \bigr) ([\widetilde {\vec{z}};\vec{q}_{\rm sys}]),
    \label{eq:multitask_network}
\end{split}
\end{equation}
with $\vec{z}\in\mathbb R^{d_z}$ the latent, $(\vec{\alpha}_{\rm m}, \vec{\beta}_{\rm m})$ the FiLM scale and shift, $[\cdot;\cdot]$ concatenation, $\widetilde {\vec{a}} \in \mathbb R^{U+1}$ raw power logits, $\widetilde {\mathbf W}_{\rm all} \in \mathbb C^{N_{\rm t} \times (U+1)}$ raw transmit beamformers whose first column is the sensing beam, and $\widetilde {\vec{u}}_{\rm s} \in \mathbb C^{N_{\rm r}}$ the raw combiner.
$\boldsymbol {\Gamma} = \{\boldsymbol {\Gamma}_{\rm e}, \boldsymbol {\Gamma}_{\rm m}, \boldsymbol {\Gamma}_{\rm p}, \boldsymbol {\Gamma}_{\rm t}, \boldsymbol {\Gamma}_{\rm r}\}$.
$M_{\boldsymbol{\Gamma}_{\rm m}}$ is a two-layer perceptron and each head of the decoder is a three-layer perceptron with Gaussian error linear unit between layers.

\subsection{Exact Enforcement of the Power and Norm Constraints}

The output layer enforces the power and norm constraints by normalization. With $b_i = \operatorname{softplus} (\widetilde a_i) + \epsilon_0>0$, $i \in \{0,1,\ldots,U\}$ and $i=0$ denoting sensing, and $a_i>0$ the power scores of \eqref{eq:residual_prediction}, \cite{Triwidyastuti_TVT_2026}
\begin{equation}
    p_i=P_{\max}a_i\Big/\textstyle\sum_{j=0}^{U}a_j,
    \qquad i=0,1,\ldots,U,
    \label{eq:power_projection_unsl}
\end{equation}
with $p_0\equiv p_{\rm s}$, so $p_i\ge0$ and $\sum_ip_i=P_{\max}$ exactly. Eq.~\eqref{eq:power_projection_unsl} is a normalization rather than a Euclidean projection, and it places every design on the face $\sum_ip_i=P_{\max}$ of \eqref{eq:isac_power_constraint}. This loses no optimality, since scaling all powers by a common factor greater than one raises every SINR in \eqref{eq:communication_sinr} and \eqref{eq:sensing_sinr}.

The heads in decoder predict a correction $\widetilde{\boldsymbol\Theta}=(\widetilde {\vec{a}}, \widetilde {\mathbf W}_{\rm all}, \widetilde {\vec{u}}_{\rm s})$ to the closed-form design $\boldsymbol\Theta^{(0)}=\{\vec{w}_i^{(0)}, \vec{u}_{\rm s}^{(0)},f_i^{(0)}\}$ with power fractions $f_i^{(0)}=p_i^{(0)}/P_{\max}$,
\begin{equation}
	a_i = b_if_i^{(0)},
	\quad
    \widehat{\vec{w}}_i = \vec{w}_i^{(0)} + \widetilde {\vec{w}}_i,
    \quad
    \widehat {\vec{u}}_{\rm s} = \vec{u}_{\rm s}^{(0)} + \widetilde {\vec{u}}_{\rm s}.
    \label{eq:residual_prediction}
\end{equation}
Zero-initializing every head output layer makes the untrained NN exactly reproduce the analytical design. The softplus contributes a common constant that cancels in \eqref{eq:power_projection_unsl}, and the corrections start at zero, so fine-tuning begins power-and-norm feasible rather than random, though the prior still violates \eqref{eq:sensing_constraint}--\eqref{eq:user_rate_constraint}. Since $\gamma^{\rm s}$ depends on $\vec{u}_{\rm s}$ only through $\rho_{\rm s}$, a learned receive correction cannot beat the exact MVDR combiner at the network's own $\mathbf R_{\rm in}$ (Remark~\ref{rem:combiner_optimality}). It is parameterized as a learned correction to keep the map differentiable, which is suboptimal once the transmit design departs from the prior. With $\widetilde {\vec{w}}_i$ column $i$ of $\widetilde {\mathbf W}_{\rm all}$,
\begin{equation}
    \vec{w}_i = \widehat {\vec{w}}_i / \|\widehat {\vec{w}}_i\|,
    \quad i = 0,1,\ldots,U,
    \qquad
    \vec{u}_{\rm s}
    = \widehat {\vec{u}}_{\rm s} / \|\widehat{\vec{u}}_{\rm s}\|,
    \label{eq:beam_projection_unsl}
\end{equation}
with $\vec{w}_0 \equiv \vec{w}_{\rm s}$, so \eqref{eq:power_projection_unsl}--\eqref{eq:beam_projection_unsl} satisfy \eqref{eq:isac_power_constraint}--\eqref{eq:combiner_constraint} for every prediction.

\subsection{Closed-Form Baseline as the Initialization Prior}
\label{subsec:model_prior}

To prevent the severe constraint violations caused by random neural network initialization, \eqref{eq:residual_prediction} instead structures the network to predict a residual correction over a low-complexity analytical prior. For communication, the RZF matrix is
\begin{equation}
	\begin{split}
		\widetilde{\mathbf W}_{\rm c}^{(0)}
		&=\mathbf H_{\rm c}^{\mathsf H}
		\left(\mathbf H_{\rm c}\mathbf H_{\rm c}^{\mathsf H}
		+\alpha_{\rm RZF}\mathbf I_U\right)^{-1},\\
		\alpha_{\rm RZF}
		&=(\alpha_0/U)\operatorname{tr}
		\left(\mathbf H_{\rm c}\mathbf H_{\rm c}^{\mathsf H}\right),
		\label{eq:rzf_initialization}
	\end{split}
\end{equation}
with columns normalized to $\vec{w}_u^{(0)}$.
Scaling the regularizer by the average channel gain through a dimensionless $\alpha_0>0$ makes the initialization invariant to absolute path loss. The sensing beam is the dominant right singular vector $\vec{w}_{\rm s}^{(0)}=\operatorname{eigvec}_{\max} (\mathbf H_{\rm T}^{\mathsf H}\mathbf H_{\rm T})$. For \eqref{eq:target_channel}, $\mathbf H_{\rm T}^{\mathsf H}\mathbf H_{\rm T} = |\alpha_{\rm T}|^2 N_{\rm r}\vec{a}_{\rm t}(\phi_{\rm T}) \vec{a}_{\rm t}^{\mathsf H}(\phi_{\rm T})$, so $\vec{w}_{\rm s}^{(0)}=\vec{a}_{\rm t}(\phi_{\rm T})/\sqrt{N_{\rm t}}$ is exactly target-matched. With $\omega_{\rm p} = \delta_{\rm p} + (1-2\delta_{\rm p})\omega$ for a power clip $0<\delta_{\rm p}<1/2$ (distinct from the dual slack $\delta$), the initial powers are $p_{\rm s}^{(0)}=(1-\omega_{\rm p})P_{\max}$ and $p_u^{(0)}=\omega_{\rm p}P_{\max}/U$, so no function starts at zero power and every stream stays in the span of the residual power head $f_u^{(0)}\ge\delta_{\rm p}/U>0$.

With initial-stage precoders $\vec{v}_i^{(0)}=\sqrt{p_i^{(0)}}\vec{w}_i^{(0)}$, the initial-stage transmit covariance is $\mathbf R_x^{(0)}=\vec{v}_{\rm s}^{(0)}(\vec{v}_{\rm s}^{(0)})^{\mathsf H}+\sum_{u=1}^{U}\vec{v}_u^{(0)}(\vec{v}_u^{(0)})^{\mathsf H}$, giving the initial-stage $\mathbf R_{\rm in}^{(0)}$ via \eqref{eq:interference_covariance} and the initial-stage interference-aware combiner as
\begin{equation}
	\vec{u}_{\rm s}^{(0)}
	=\vec{y}^{(0)}/\|\vec{y}^{(0)}\|,
	\qquad
	\vec{y}^{(0)}=(\mathbf R_{\rm in}^{(0)})^{-1}
	\mathbf H_{\rm T}\vec{w}_{\rm s}^{(0)}.
	\label{eq:mvdr_initialization}
\end{equation}
Although \eqref{eq:mvdr_initialization} omits the leakage $B$ in the denominator of \eqref{eq:sensing_sinr}, it is exactly optimal for the prior's own transmit design.

\begin{remark}[Optimality of the initial sensing combiner]
	\label{rem:combiner_optimality}
	By \eqref{eq:sensing_sinr_closed}, $\gamma^{\rm s}$ depends on $\vec{u}_{\rm s}$ only through $\rho_{\rm s}$ and is strictly increasing in it whenever $A>0$, with maximizer $\vec{u}_{\rm s}\propto\mathbf R_{\rm in}^{-1}\vec{a}_{\rm r}$. Since $\mathbf H_{\rm T}\vec{w}_{\rm s}^{(0)}\propto\vec{a}_{\rm r}$, that maximizer is \eqref{eq:mvdr_initialization}, which therefore maximizes $\gamma^{\rm s}$ for the transmit design that generated $\mathbf R_{\rm in}^{(0)}$. Once the network changes the powers or beams, $\mathbf R_{\rm in}$ changes and \eqref{eq:mvdr_initialization} is no longer the exact maximizer. $B$ affects the attainable value but not the optimal combining direction, so omitting it is 	lossless.
\end{remark}

\subsection{Physics-Defined Loss with Calibrated Constraints}

Substituting the normalized outputs into \eqref{eq:communication_sinr} and \eqref{eq:sensing_sinr} makes both utilities differentiable in $\boldsymbol\Gamma$. The constraints \eqref{eq:sensing_constraint}--\eqref{eq:user_rate_constraint} become $\bar g_i\leq0$ through the normalized residuals
\begin{equation}
	\bar g_{\rm s} = \frac{\Gamma_{\rm s}^{\min}-\gamma^{\rm s}}{\Gamma_{\rm s}^{\min}+\epsilon_0} , \qquad
	\bar g_u = \frac{R_u^{\min}-R_u}{R_u^{\min}+\epsilon_0},
	\label{eq:normalized_residuals}
\end{equation}
with violation $g_i=[\bar g_i]_+$ (not to be confused with the fading vector $\vec{g}_u$).

Constrained learning relaxes the per-realization requirements to $\mathbb E\{g_i\}\leq\varepsilon_i(\omega)$, $i\in\mathcal I=\{{\rm s},1,\ldots,U\}$~\cite{Chamon_TIT_2023,Eisen_TSP_2019}. Since \eqref{eq:isac_problem} is infeasible on some realizations, any $\varepsilon_i$ below the attainable floor leaves $\mathbb E\{g_i\}>\varepsilon_i$ at every iterate, and the multipliers diverge. For sensing, this is certified by Corollary~\ref{cor:range}: every design has $\mathbb E\{g_{\rm s}\}\ge\varepsilon_{\rm NL}$, so a target $\varepsilon_{\rm s}<\varepsilon_{\rm NL}$ raises the multiplier by at least $\eta_\mu(\varepsilon_{\rm NL}-\varepsilon_{\rm s})>0$ per ascent step in expectation.

We therefore measure the floor. The priority range $[\omega_{\min},\omega_{\max}]$ is split into $K$ bins centered at $\omega_k$, and in each bin the sensing violation is evaluated through \eqref{eq:sensing_sinr} at the sensing-only allocation with the baseline beams and combiner,
\begin{equation}
    \check g_{\rm s}(\omega_k)
    =\mathbb E\{g_{\rm s}\}\big|_{\mathcal A(\omega_k)},
    \label{eq:violation_floor}
\end{equation}
where $\mathcal A(\omega_k)$ denotes that allocation. By Theorem~\ref{thm:max_sensing}, it minimizes $\mathbb E\{g_{\rm s}\}$ over the power splits for these beams and combiner, and its mild $\omega$-dependence comes only from the combiner, which is built before the split. Unlike the analytic floor \eqref{eq:g_floor_analytic}, the measured floor therefore varies mildly with $\omega_k$, which the per-bin calibration absorbs. Evaluating the same setup at $\omega=0.5$ gives the per-realization floor $\check g_{\rm s}^{(b)}$, and a slack $\delta>0$ is reserved for the sensing dual.

The loss applies a Powell--Hestenes--Rockafellar (PHR) penalty to the clipped violations, with multiplier ascent~\cite{Tian_JIOT_2023,Park_AAAI_2023}, a geometric penalty schedule capped at $\kappa_{\max}$~\cite{Park_AAAI_2023}, and, unlike these works, a box on the multipliers whose cap follows the objective scale. Two refinements are added. First, the multipliers are carried per priority bin as $\mu_i(\omega_k)$, and each realization is mapped to bin index

\begin{equation}
    k(\omega) = \min \left\{ \left \lfloor K\, \frac {\omega - \omega_{\min}} {\omega_{\max} - \omega_{\min}} \right \rfloor + 1,\; K \right\}
    \label{eq:omega_bin}
\end{equation}
since the difficulty of the sensing constraint varies strongly with $\omega$. Second, only the violation above the measured floor is charged, through the achievability gap
\begin{equation}
    \tilde g_{\rm s}^{(b)}
    =\bigl[g_{\rm s}^{(b)}-\check g_{\rm s}^{(b)}\bigr]_+,
    \label{eq:achievability_gap}
\end{equation}
which vanishes iff the allocation reaches the floor, while worse allocations on infeasible draws keep a gradient toward it.

The sensing term is thus the PHR penalty of $g_{\rm s}\le\check g_{\rm s}^{(b)}$, which reduces to \eqref{eq:sensing_constraint} where $\check g_{\rm s}^{(b)}=0$, and the loss reads
\begin{equation}
\begin{split}
    \mathcal L_{\rm PI\text{-}UnSL}
    ={}&-\omega\widetilde{\mathcal C}
    -(1-\omega)\widetilde{\mathcal S}
    +\mu_{\rm s}(\omega_{k})\tilde g_{\rm s}
    +(\kappa/2)\tilde g_{\rm s}^2
    \\
    &+\sum_{u=1}^{U}
    \left(\mu_u(\omega_{k})g_u+(\kappa/2)g_u^2\right)
    \\
    &+\lambda_{\rm phy} \mathcal{L}_{\rm phy} + \lambda_{\rm L1}\mathcal L_{\rm L1},
    \label{eq:unsl_augmented_loss}
\end{split}
\end{equation}
with multipliers $\mu_i\geq0$, penalty coefficient $\kappa>0$, and the sample index $(b)$ suppressed.

Each term is the active branch of the PHR penalty on a clipped violation, so a constraint met with slack contributes no gradient, and the normalizations make power and norm penalties unnecessary. Because the square is taken per sample, the average penalizes $\mathbb E\{g_i^2\}=(\mathbb E\{g_i\})^2+\operatorname{var}(g_i)$ and hence also the spread of the violations, as a per-realization requirement demands.

$\mathcal L_{\rm phy}$ regularizes the latent toward what $\vec{q}_{\rm sys}$ alone predicts,
\begin{equation}
\begin{split}
\mathcal{L}_{\rm phy} = & D_{\rm KL} \! \left( \mathcal N (\vec{\mu}_x, \operatorname{diag} \vec{\sigma}_x^{2}) \, \middle \| \, \mathcal N (\vec{\mu}_q, \operatorname{diag} \vec{\sigma}_q^{2}) \right) \\
= & \frac{1} {2d_z} \sum_{d=1}^{d_z} \left[ \log \frac{ \sigma_{q,d}^{2}} {\sigma_{x,d}^{2}} + \frac{ \sigma_{x,d}^{2} + (\mu_{x,d} - \mu_{q,d})^{2}} {\sigma_{q,d}^{2}} -1 \right],
\end{split}
\end{equation}
where the encoder in~\eqref{eq:multitask_network} is made variational, $\vec{z}\sim\mathcal N(\vec{\mu}_x,\operatorname{diag}\vec{\sigma}_x^2)=E_{\boldsymbol\Gamma_{\rm e}}(\mathcal X_{\rm in})$, sampled by reparameterization~\cite{Kingma_ICLR_2014}, and a second branch $E_{\boldsymbol\Gamma_{\rm q}}(\vec{q}_{\rm sys})$ gives $(\vec{\mu}_q,\vec{\sigma}_q^2)$. Unlike a variational autoencoder, no reconstruction is trained, and the latent only couples the encoder to $\vec{q}_{\rm sys}$. Our companion work uses the same machinery for channel estimation~\cite{Nguyen_TGCN_2026}.

$\mathcal{L}_{\rm L1}=\frac{1}{|\mathcal F|}\sum_{f\in\mathcal F}\frac{1}{|\vec{q}_{\rm sys}|}\sum_{j=1}^{|\vec{q}_{\rm sys}|} s_j\big\|\vec{w}_{j}^{f}\big\|_2$, with $\mathcal F=\{M_{\boldsymbol\Gamma_{\rm m}},D_{\boldsymbol\Gamma_{\rm p}},D_{\boldsymbol\Gamma_{\rm t}},D_{\boldsymbol\Gamma_{\rm r}},E_{\boldsymbol\Gamma_{\rm q}}\}$, measures which $\vec{q}_{\rm sys}$ features the trained modules use, where $\vec{w}_j^f$ is the input-layer weight column of module $f$ for entry $j$ and $s_j$ is the standard deviation of weighted entry $j$, estimated on training batches. We refer to the closed-form statistics in $\vec{q}_{\rm sys}$, their weighting, $\mathcal L_{\rm phy}$, and $\mathcal L_{\rm L1}$ together as the physics conditioning.

The parameters are updated by Adam with rate $\eta_{\boldsymbol\Gamma}$ on the mini-batch mean, $\boldsymbol{\Gamma} \leftarrow \boldsymbol{\Gamma} - \eta_{\boldsymbol {\Gamma}} \nabla_{ \boldsymbol{\Gamma}} \bigl( |\mathcal B|^{-1} \sum_{b\in\mathcal B} \mathcal L_{\rm PI \text{-}UnSL}^{(b)} \bigr)$, and the multipliers ascend once per epoch in each bin over its epoch sample set $\mathcal E_k$,
\begin{equation}
\begin{split}
    \mu_{\rm s}(\omega_k)&\leftarrow
    \min\left\{
    \left[\mu_{\rm s}(\omega_k)+\eta_{\mu}
    \left(\tilde g_{\rm s}^{\rm avg}-\delta\right)\right]_+,
    \bar\mu\right\},
    \\
    \mu_u(\omega_k)&\leftarrow
    \min\left\{
    \left[\mu_u(\omega_k)+\eta_{\mu}
    \left(\mathrm{CVaR}_\alpha^{k}\{g_u\}-\varepsilon_u\right)\right]_+,
    \bar\mu\right\},
    \label{eq:dual_update}
\end{split}
\end{equation}
where $\tilde g_{\rm s}^{\rm avg} = |\mathcal E_k|^{-1} \sum_{b \in \mathcal E_k} \tilde g_{\rm s}^{(b)}$ and the per-user tail statistic is the sample conditional value at risk
\begin{equation}
\begin{split}
    \mathrm{CVaR}_\alpha^{k}\{g_u\}
    &=\frac{1}{m_k}\sum_{b\in\mathcal T_k}g_u^{(b)},\\
    m_k&=\max\big\{1,\operatorname{round}(\alpha|\mathcal E_k|)\big\},
    \label{eq:cvar}
\end{split}
\end{equation}
with $\mathcal T_k$ the $m_k$ realizations of largest $g_u^{(b)}$ in bin $k$, per user.
The sensing dual ascends on the same gap that the penalty charges, with slack target $\delta$. The user dual ascends on the CVaR tail instead, since a mean-tracking dual is blind to rare severe violations, and by the dual representation of CVaR~\cite{Chow_NIPS_2015} it tracks the worst-case distribution in an uncertainty set, like an outage constraint without a closed-form bound. Its target is $\varepsilon_u=\check g_u+\delta_u$, with $\check g_u$ the CVaR at the communication-only extreme. Both updates use clipped violations, since signed residuals would let easy realizations pin the multipliers at zero.

The multiplier cap follows the objective scale,
\begin{equation}
    \bar\mu=\min\Big\{\mu_{\max},\;
    \varrho\,\widehat{\mathcal U}\big/\bar\varepsilon\Big\},
    \qquad
    \bar\varepsilon=\frac{1}{K}\sum_{k=1}^{K}
    \big[\check g_{\rm s}(\omega_k)+\delta\big],
    \label{eq:multiplier_cap}
\end{equation}
where $\widehat{\mathcal U}$ is the mean joint utility of the latest training epoch, $\varrho\in(0,1)$ bounds the penalty as a fraction of the objective, $\mu_{\max}$ is a loose absolute safeguard, and $\bar\varepsilon$ inherits the calibration. The penalty grows geometrically via
\begin{equation}
\begin{split}
    &\kappa\leftarrow\min\{\nu\kappa,\kappa_{\max}\}\quad\text{if}\\
    &\max\{\tilde g_{\rm s}^{\rm val},\, g_{\rm c}^{\max}\}
    >\varsigma\max\{\tilde g_{\rm s}^{\rm val,prev},\, g_{\rm c}^{\max,\rm prev}\},
\end{split}
\label{eq:penalty_update}
\end{equation}
with $\nu>1$ and $0<\varsigma<1$, where $\tilde g_{\rm s}^{\rm val}$ is the mean validation gap and $g_{\rm c}^{\max}=\max_{b,u}g_u^{(b)}$ is the largest validation user violation. If this statistic does not improve on its best value by a relative $2\%$ for $T_{\rm f}$ consecutive epochs, the dual freezes and $\kappa$ stops growing. This freeze rule is a plateau detector that stops further multiplier growth.
The cap $\kappa_{\max}$ keeps the $O(1)$ utility visible to $\omega$, since otherwise the penalty dominates and the network collapses onto a single operating point.

\begin{remark}[Physics in the loss]
\label{rem:physics_loss}
The physics enters \eqref{eq:unsl_augmented_loss} at three levels: the utilities and violations are evaluated through the SINR model \eqref{eq:communication_sinr}, \eqref{eq:sensing_sinr}, the dual targets are calibrated to the floor that Theorem~\ref{thm:max_sensing} makes measurable, and $\mathcal L_{\rm phy}$ ties the latent to closed-form statistics. Because every dual target is attainable by construction, the loss remains a bounded augmented Lagrangian rather than degenerating into a divergent penalty.
\end{remark}

\subsection{Training Procedure and Calibrated Dual Properties}
\label{subsec:convergence}

\begin{algorithm}[!t]
\footnotesize
\SetCommentSty{algcomment}
\caption{PI-UnSL Training and One-Pass Inference (physics-informed steps in green)}
\label{alg:mi_unsl}
\KwIn{Channel generator $\mathcal G$ of Section~\ref{sec:system_model}; epochs $N_{\rm unsl}$; rates $\eta_{\boldsymbol{\Gamma}},\eta_\mu$; slacks $\delta,\delta_u$; CVaR level $\alpha$; cap fraction $\varrho$; penalty $\kappa\in[\kappa_{\min},\kappa_{\max}]$, $\nu>1$, $\varsigma$; patience $T_{\rm f}$; QoS target $\varrho_{\rm QoS}$; $K$ bins over $[\omega_{\min},\omega_{\max}]$.}
\KwOut{Trained $\boldsymbol{\Gamma}$ and a one-pass ISAC mapping.}
\tcp{Offline calibration}
\pinov{Evaluate the baseline at $\omega\in\{0,1\}$ to set $\mathcal C^{\rm ref},\mathcal S^{\rm ref}$}\tcp*[r]{\pinov{utility scales}}
\pinov{Measure $\check g_{\rm s}(\omega_k)$ and $\check g_{\rm s}^{(b)}$ at the sensing-only allocation by \eqref{eq:violation_floor}; set $\varepsilon_u=\check g_u+\delta_u$ and $\bar\varepsilon$}\tcp*[r]{\pinov{Thm.~\ref{thm:max_sensing} floor}}
\pinov{Zero-initialize the head output layers, so the untrained mapping equals the baseline}\tcp*[r]{\pinov{no labels}}
Set $\mu_{\rm s}(\omega_k)=\mu_u(\omega_k)=0$ and $\kappa=\kappa_{\min}$\;
\tcp{Label-free training}
\For{$e=1,\ldots,N_{\rm unsl}$}{
  \For{each mini-batch $\mathcal B\sim\mathcal G$}{
    Sample $\omega\sim\mathcal U[\omega_{\min},\omega_{\max}]$ per realization and form $\mathcal X_{\rm in}$\;
    \pinov{Build $\{\vec{w}_i^{(0)},\vec{u}_{\rm s}^{(0)},f_i^{(0)}\}$ by \eqref{eq:rzf_initialization}--\eqref{eq:mvdr_initialization}}\tcp*[r]{\pinov{closed-form baseline}}
    \pinov{Compute $\mathbf G_{\rm comm},g_{\rm sens},\vec{l}_{\rm sens},\mu_{\rm feas}$ at the baseline and form $\vec{w}_{\rm sys}\odot\vec{q}_{\rm sys}$ by \eqref{eq:qsys}}\tcp*[r]{\pinov{physics statistics}}
    Forward pass \eqref{eq:multitask_network}: encoder, FiLM conditioning, and the three heads\;
    \pinov{Correct the baseline by \eqref{eq:residual_prediction} and normalize by \eqref{eq:power_projection_unsl}--\eqref{eq:beam_projection_unsl}}\tcp*[r]{\pinov{exact normalization}}
    \pinov{Evaluate $R_u$ and $\gamma^{\rm s}$ through \eqref{eq:communication_sinr}, \eqref{eq:sensing_sinr}; form $g_i$ and $\tilde g_{\rm s}$ by \eqref{eq:normalized_residuals}, \eqref{eq:achievability_gap}}\tcp*[r]{\pinov{physics loss}}
    Update $\boldsymbol{\Gamma}$ by Adam on the batch mean of \eqref{eq:unsl_augmented_loss}, including $\mathcal L_{\rm phy}$ and $\mathcal L_{\rm L1}$\;
  }
  \If{dual not frozen}{
    \pinov{Ascend $\mu_{\rm s}(\omega_k)$ on the mean gap and $\mu_u(\omega_k)$ on the CVaR tail by \eqref{eq:dual_update}, \eqref{eq:cvar}, capped by \eqref{eq:multiplier_cap}}\tcp*[r]{\pinov{calibrated dual}}
    Grow $\kappa$ by \eqref{eq:penalty_update}\;
  }
  \If{no improvement for $T_{\rm f}$ epochs}{
    Freeze the dual and hold $\kappa$\;
  }
  Keep the checkpoint of maximum validation utility with joint QoS $\ge\varrho_{\rm QoS}$ (of maximum joint QoS if none qualifies); stop early if it no longer improves\;
}
\tcp{One-pass inference: no labels and no iterations}
For the current CSI and requested $\omega$, \pinov{build the baseline and $\vec{q}_{\rm sys}$}, run one forward pass, normalize, and return $\{p_{\rm s},p_u,\vec{w}_{\rm s},\vec{w}_u,\vec{u}_{\rm s}\}$ meeting \eqref{eq:isac_power_constraint}--\eqref{eq:combiner_constraint}\;
\end{algorithm}

Algorithm~\ref{alg:mi_unsl} summarizes the procedure, and Fig.~\ref{fig:pipeline} illustrates it. The physics-informed steps, which carry the contributions of this work, are highlighted in green. The zero-initialized heads make the first pass identical to the analytical design, after which the network is fine-tuned label-free on \eqref{eq:unsl_augmented_loss}.

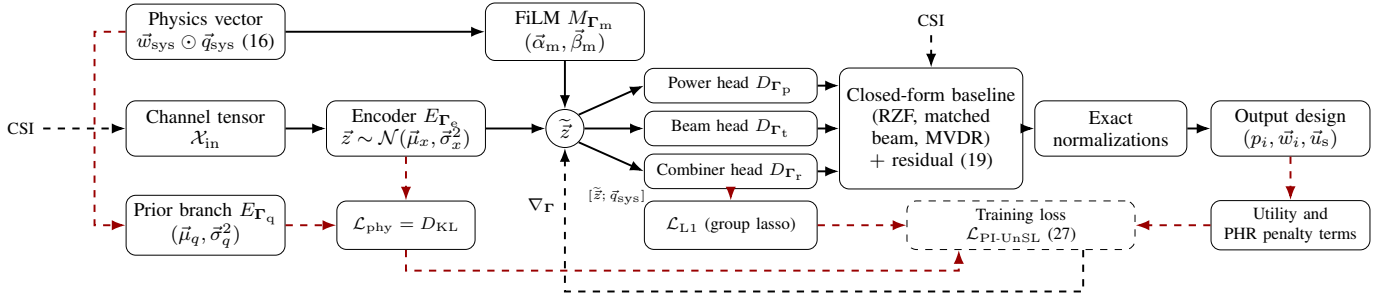
\begin{figure*}[!t]
    \centering
    \resizebox{\textwidth}{!}{%
    \begin{tikzpicture}[
        block/.style={draw, rounded corners, align=center,
                      minimum height=0.8cm, font=\footnotesize, inner sep=3pt},
        head/.style={draw, rounded corners, align=center, minimum width=2.5cm,
                      minimum height=0.5cm, font=\scriptsize, inner sep=2pt},
        lossblock/.style={draw, rounded corners, align=center,
                      minimum height=0.7cm, font=\scriptsize, inner sep=3pt},
        sum/.style={draw, circle, minimum size=0.55cm, font=\footnotesize, inner sep=0pt},
        arr/.style={-{Latex[length=2mm]}, thick},
        darr/.style={-{Latex[length=2mm]}, thick, dashed, red!60!black},
        garr/.style={-{Latex[length=2mm]}, thick, dashed},
    ]
    \node[block, minimum width=2.3cm] (xin) at (0,0) {Channel tensor\\ $\mathcal X_{\rm in}$};
    \node[block, minimum width=2.3cm] (enc) at (2.9,0)
        {Encoder $E_{\boldsymbol\Gamma_{\rm e}}$\\ $\vec{z}\sim\mathcal N(\vec{\mu}_x,\vec{\sigma}_x^2)$};
    \node[sum] (mod) at (5.2,0) {$\widetilde{\vec{z}}$};
    \node[head] (hp) at (7.6,0.62) {Power head $D_{\boldsymbol\Gamma_{\rm p}}$};
    \node[head] (ht) at (7.6,0)    {Beam head $D_{\boldsymbol\Gamma_{\rm t}}$};
    \node[head] (hr) at (7.6,-0.62) {Combiner head $D_{\boldsymbol\Gamma_{\rm r}}$};
    \node[block, minimum width=2.2cm, minimum height=1.75cm] (prior) at (10.5,0)
        {Closed-form baseline\\ (RZF, matched\\ beam, MVDR)\\ $+$ residual~\eqref{eq:residual_prediction}};
    \node[block, minimum width=2.2cm] (proj) at (13.1,0) {Exact\\ normalizations};
    \node[block, minimum width=2.3cm] (out) at (15.7,0)
        {Output design\\ $(p_i,\vec{w}_i,\vec{u}_{\rm s})$};

    \draw[arr] (xin) -- (enc);
    \draw[arr] (enc) -- (mod);
    \draw[arr] (mod.north east) -- (hp.west);
    \draw[arr] (mod) -- (ht);
    \draw[arr] (mod.south east) -- (hr.west);
    \node[font=\tiny, anchor=north] at (5.95,-0.72) {$[\widetilde{\vec{z}};\vec{q}_{\rm sys}]$};
    \draw[arr] (hp.east) -- (hp.east -| prior.west);
    \draw[arr] (ht.east) -- (ht.east -| prior.west);
    \draw[arr] (hr.east) -- (hr.east -| prior.west);
    \draw[arr] (prior) -- (proj);
    \draw[arr] (proj) -- (out);
    \draw[arr, dashed] (10.5,1.35) node[above, font=\scriptsize] {CSI}
        -- (prior.north);
    \draw[arr, dashed] (-2.35,0) node[left, font=\scriptsize] {CSI}
        -- (xin.west);

    \node[block, minimum width=2.3cm] (qsys) at (0,1.4)
        {Physics vector\\ $\vec{w}_{\rm sys}\odot\vec{q}_{\rm sys}$~\eqref{eq:qsys}};
    \node[block, minimum width=2.3cm] (film) at (5.2,1.4)
        {FiLM $M_{\boldsymbol\Gamma_{\rm m}}$\\ $(\vec{\alpha}_{\rm m},\vec{\beta}_{\rm m})$};
    \draw[arr] (qsys) -- (film);
    \draw[arr] (film) -- (mod);

    \node[block, minimum width=2.3cm] (prq) at (0,-1.4)
        {Prior branch $E_{\boldsymbol\Gamma_{\rm q}}$\\ $(\vec{\mu}_q,\vec{\sigma}_q^2)$};
    \node[lossblock, minimum width=2.0cm] (kl) at (2.9,-1.4) {$\mathcal L_{\rm phy}=D_{\rm KL}$};
    \node[lossblock, minimum width=2.5cm] (l1) at (7.6,-1.4)
        {$\mathcal L_{\rm L1}$ (group lasso)};
    \node[lossblock, minimum width=3.3cm, dashed] (loss) at (11.8,-1.4)
        {Training loss\\ $\mathcal L_{\rm PI\text{-}UnSL}$~\eqref{eq:unsl_augmented_loss}};
    \node[lossblock, minimum width=2.3cm] (phr) at (15.7,-1.4)
        {Utility and\\ PHR penalty terms};

    \draw[darr] (qsys.west) -- ++(-0.45,0) |- (prq.west);
    \draw[darr] (prq) -- (kl);
    \draw[darr] (enc) -- (kl);
    \draw[darr] (hr) -- (l1);
    \draw[darr] (out) -- (phr);
    \draw[darr] (phr) -- (loss);
    \draw[darr] (l1) -- (loss);
    \draw[darr] (kl.south) -- ++(0,-0.3) -| ([xshift=-0.9cm]loss.south);
    \draw[garr] ([xshift=0.9cm]loss.south) -- ++(0,-0.6) -| (mod.south)
        node[pos=0.8, left, font=\scriptsize] {$\nabla_{\boldsymbol\Gamma}$};
    \end{tikzpicture}%
    }
    \caption{PI-UnSL forward pass (solid arrows) and training loop; dashed red arrows feed the losses, and the dashed black arrow back-propagates them into $\boldsymbol\Gamma$.}
    \label{fig:pipeline}
\end{figure*}

Two properties of the training regime follow, neither of which claims global optimality of the NP-hard problem \eqref{eq:isac_problem}. First, the projection in \eqref{eq:dual_update} keeps every multiplier in $[0,\bar\mu]$. With $\mathbb E_k\{\cdot\}$ the expectation over bin $k$, if the bin residuals $\Upsilon_{{\rm s},k}(\vec{\mu})=\mathbb E_k\{\tilde g_{\rm s}\}$ and $\Upsilon_{u,k}(\vec{\mu})=\mathrm{CVaR}_\alpha^k\{g_u\}$ depend continuously on the multipliers, the capped dual map has a fixed point in $[0,\bar\mu]^{K(1+U)}$. There, each multiplier is zero, balances its target, or sits at the cap, so an unmeetable constraint is caught at the cap instead of diverging. In the reported runs the dual freezes early (Section~\ref{subsec:results}), so this argument bounds the multipliers rather than describing most of training.

Second, the gap penalty reduces to the standard PHR penalty of \eqref{eq:sensing_constraint} wherever the constraint is attainable ($\check g_{\rm s}^{(b)}=0$), and on an infeasible realization it vanishes only at the floor, elsewhere keeping the gradient $-\nabla\gamma^{\rm s}/(\Gamma_{\rm s}^{\min}+\epsilon_0)$ toward it. Both properties concern training stability, not outcome quality. The reported mapping, the checkpoint of maximum validation utility among the epochs meeting $\varrho_{\rm QoS}$, comes after the dual freezes, and its gap to an iterative primal--dual solver is measured in Section~\ref{sec:results}.

\section{Numerical Results and Discussion}
\label{sec:results}

\begin{table}[!tb]
	\centering
	\caption{Simulation Parameters}
	\label{tab:sim_params}
	\footnotesize
	\begin{tabular}{@{}ll@{}}
		\hline
		Parameter & Value \\
		\hline
		$N_{\rm t}=N_{\rm r}$, $U$, $L_{\rm cl}$ & $8$, $4$, $4$ \\
		Carrier, bandwidth & $28$~GHz, $100$~MHz \\
		$P_{\max}$, $T_{\rm b}$, $G_{\rm p}=N_{\rm s}$ & $30$~dBm, $100~\mu$s, $256$ \\
		Element gains $G_{\rm t}=G_{\rm r}$ & $1$ \\
		Comm.\ / sensing noise figure & $7$~dB / $10$~dB \\
		Residual SI, Rician $K_u$, spatial corr. & $-110$~dB, $5$~dB, $0.55$ \\
		Path-loss exponent, shadowing & $2.2$, $4$~dB \\
		User range, speed, sector & $25$--$120$~m, $\le 30$~m/s, $\pm65^\circ$ \\
		Target range, speed, sector & $20$--$100$~m, $\le 35$~m/s, $\pm60^\circ$ \\
		Target RCS median (mean), std & $10$ ($12.7$)~m$^2$, $3$~dB \\
		Clutter range, RCS, speed std & $10$--$150$~m, $0.05$--$1$~m$^2$, $2$~m/s \\
		RZF regularization $\alpha_0$ & $0.01$ \\
		$\Gamma_{\rm s}^{\min}$, $R_u^{\min}$ & $12$~dB, $0.25$~bit/s/Hz \\
		Power clip $\delta_{\rm p}$, dual slack $\delta$ & $0.05$, $0.01$ \\
		\hline
	\end{tabular}
\end{table}

Simulations use the parameters in Table~\ref{tab:sim_params}, with $2\,500$ tracks of $8$ consecutive blocks ($0.8$~ms per track), which are split by track into $15\,000$ training, $2\,496$ validation, and $2\,504$ testing blocks. The $G_{\rm p}=N_{\rm s}=256$ samples occupy $2.56~\mu$s at $100$~MHz, so $N_{\rm s}T_{\rm s}\ll T_{\rm b}$.
The track structure serves the split and the outdated-CSI evaluation. The mobility evaluation uses longer tracks, described below.
Fig.~\ref{fig:isac_topology} shows one test realization, where PI-UnSL keeps its main lobe on the target while lifting the pattern toward the users.

\begin{figure}[!tb]
    \centering
    \includegraphics[width=\columnwidth]{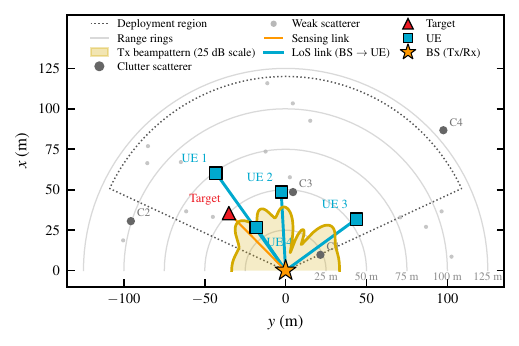}
    \caption{Example test scene with the normalized PI-UnSL transmit beampattern at $\omega=0.45$ ($25$~dB scale); weak scatterers are illustrative only.}
    \label{fig:isac_topology}
\end{figure}

\begin{figure}[!tb]
	\centering
	\includegraphics[width=\columnwidth]{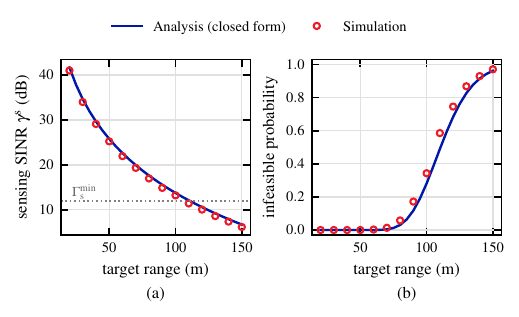}
	\caption{Closed form against simulation at the sensing-only allocation. (a) $\gamma_{\rm NL}$ and sensing SINR. (b) Infeasible probability.}
	\label{fig:validation}
\end{figure}

\begin{figure}[!t]
	\centering
	\includegraphics[width=\columnwidth]{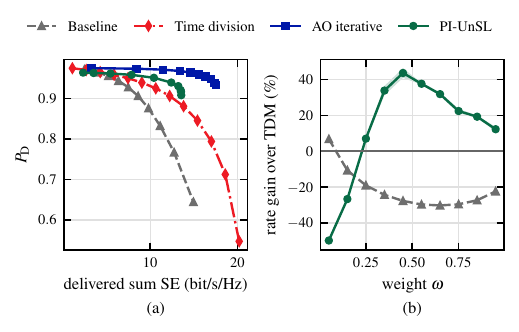}
	\caption{PI-UnSL against time-division ISAC at $P_{\rm FA}=10^{-4}$. (a) $P_{\rm D}$ against delivered sum SE. (b) Rate gain at matched $P_{\rm D}$ (AO omitted).}
	\label{fig:methods}
\end{figure}

\begin{figure}[!tb]
	\centering
	\includegraphics[width=\columnwidth]{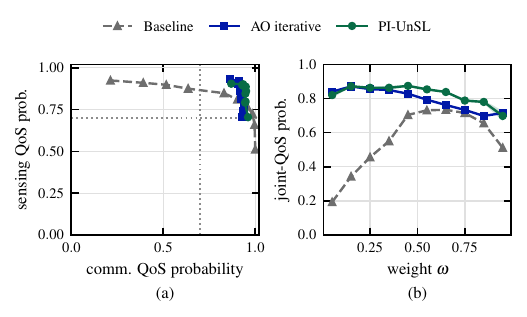}
	\caption{Dual-function performance. (a) Sensing vs. communication QoS over $\omega$. (b) Probability one realization satisfies all requirements at once.}
	\label{fig:tradeoff}
\end{figure}

\begin{figure}[!tb]
	\centering
	\includegraphics[width=.8\columnwidth]{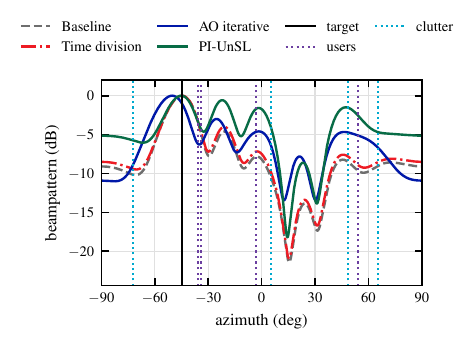}
	\caption{Normalized transmit beampattern at $\omega=0.45$ on one test realization. Time division shows its duty-cycle-averaged pattern at $\rho_{\rm td}=0.5$.}
	\label{fig:beamforming}
\end{figure}

\begin{figure*}[!tb]
	\centering
	\includegraphics[width=2\columnwidth]{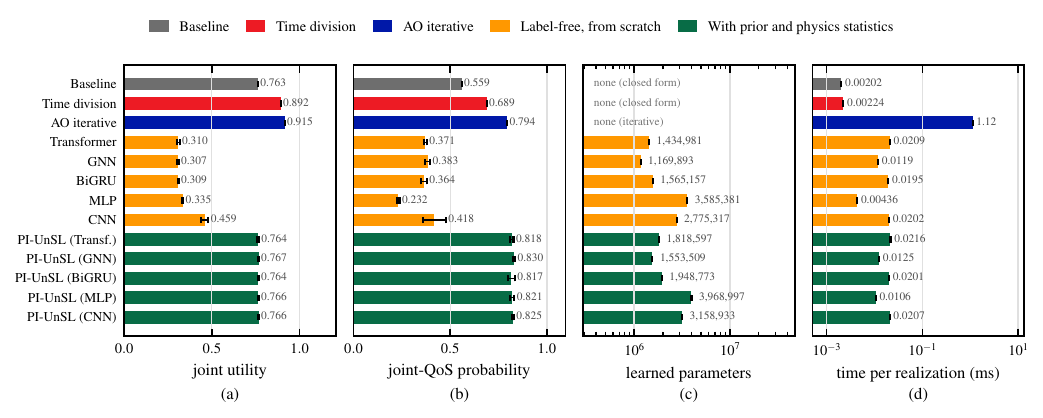}
	\caption{PI-UnSL against label-free networks trained from scratch and with the prior, over three seeds: (a) joint utility, (b) joint-QoS probability, (c) learned parameters, (d) time per realization. Time division is scored per block at $\rho_{\rm td}=\omega_k$.}
	\label{fig:architectures}
\end{figure*}

\begin{figure}[!t]
	\centering
	\includegraphics[width=\columnwidth]{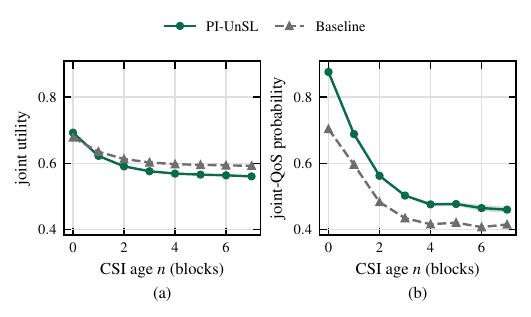}
	\caption{PI-UnSL and the baseline designed on CSI that is $n$ blocks old, at $\omega=0.45$: (a) joint utility, (b) joint-QoS probability.}
	\label{fig:outdated}
	\medskip
	\includegraphics[width=\columnwidth]{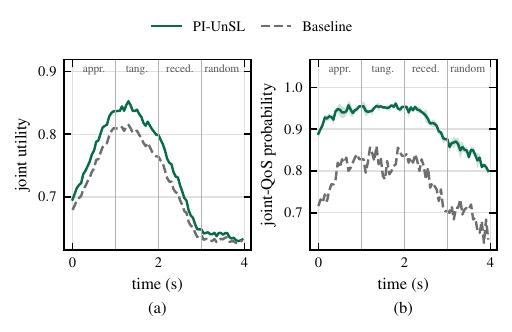}
	\caption{PI-UnSL and the baseline over four heading segments of $1$~s, per-block CSI, $\omega=0.45$: (a) joint utility, (b) joint-QoS probability.}
	\label{fig:mobility}
\end{figure}

\begin{table}[!tb]
	\centering
	\caption{Measured Floor and Infeasible Fraction}
	\label{tab:floor}
	\footnotesize
	\begin{tabular}{@{}cccc@{}}
		\hline
		$\omega_k$ & Floor $\check g_{\rm s}(\omega_k)$ & Infeasible & Floor $+\delta$ \\
		\hline
		$0.095$ & $0.014$ & $4.0\%$ & $0.024$ \\
		$0.365$ & $0.015$ & $4.0\%$ & $0.025$ \\
		$0.635$ & $0.016$ & $4.5\%$ & $0.026$ \\
		$0.905$ & $0.019$ & $6.1\%$ & $0.029$ \\
		\hline
		\multicolumn{4}{p{165pt}}{Bins $1,4,7,10$ of $K=10$, with $800$ validation realizations per bin, and $\bar\varepsilon$ averages all ten. Last column is the cap scale, not the dual target.}
	\end{tabular}
\end{table}

Table~\ref{tab:floor} reports the measured sensing violation floor. It is small and rises only mildly with $\omega$, through the combiner, and only a small fraction of realizations is infeasible. Since some of them fail only because of the fixed baseline beams (the matched beam of Theorem~\ref{thm:max_sensing} need not be optimal), the measured floor is a safe upper bound for calibrating the dual. The binding limit is $P_{\max}$ under RCS fading, not a clutter ceiling.

Fig.~\ref{fig:validation} validates the analysis of Section~\ref{subsec:feasibility} against Monte Carlo simulation at the sensing-only allocation, with clutter and residual SI. The closed-form bound \eqref{eq:noise_limited} closely tracks and slightly exceeds the simulated sensing SINR at every range, and the infeasible probability predicted by Corollary~\ref{cor:range} follows the simulated trend. The sensing constraint becomes infeasible only for distant targets, which produces the floor of Table~\ref{tab:floor}.

We also compare three sensing-multiplier rules in a single-seed run of a reduced configuration, with the physics statistics but without $\mathcal L_{\rm phy}$, $\mathcal L_{\rm L1}$, and the $\vec{q}_{\rm sys}$ weighting, and with the mean rather than the largest validation user violation in \eqref{eq:penalty_update}. With the target below the floor and the safeguards removed, the multiplier grows linearly without bound, whereas with the calibrated target and the safeguards it freezes well inside its cap. Their performance is nearly identical, so the calibration acts as a stability safeguard rather than a source of performance.

\label{subsec:results}

Training uses $15\,000$ unlabeled snapshots, $\omega \sim \mathcal U[0.05,0.95]$, Adam at $\eta_{\boldsymbol \Gamma} = 10^{-3}$ on a cosine schedule, batch $256$, $\kappa \in [5,50]$, $\eta_\mu = 0.5$, $K = 10$, $\delta_u = 0.05$, $\alpha = 0.10$, $\varrho = 0.30$, $T_{\rm f} = 8$, $\nu = 1.5$, $\varsigma = 0.9$, $\mu_{\max} = 100$, and a checkpoint maximizing validation utility among epochs with joint QoS $\ge0.85$. When no epoch reaches this target, as for every blank-slate network below, the epoch of maximum validation joint QoS is kept.
Additionally, the reported NN carries the physics-consistency term $\mathcal L_{\rm phy}$ at $\lambda_{\rm phy}=0.001$, ramped linearly from zero over the first $20$ epochs, and the group-lasso term $\mathcal L_{\rm L1}$ at $\lambda_{\rm L1}=0.1$ of \eqref{eq:unsl_augmented_loss}, together with the weighted $\vec{q}_{\rm sys}$ of \eqref{eq:qsys}. Training runs for at most $300$ epochs, with early stopping after $40$ epochs without improvement.
Under this schedule, $\kappa$ reaches $\kappa_{\max}$ within about ten epochs and the dual freezes soon after, with multipliers far below their cap. The selected checkpoints all come after the freeze, so the reported network is trained mostly under a fixed penalty after a short dual warm-up.

The baseline is the closed-form design of Section~\ref{subsec:model_prior}, that is, RZF user beams, the target-matched sensing beam, the interference-aware combiner, and the priority-dependent power split with equal user powers, which is also the untrained PI-UnSL output.
We also benchmark against time-division multiplexing (TDM) ISAC, which avoids interference by dedicating a fraction $\rho_{\rm td}$ of the coherence block to full-power communication (scaling data rate by $\rho_{\rm td}$) and $1-\rho_{\rm td}$ to full-power sensing (reducing coherent gain to $(1-\rho_{\rm td}) G_{\rm p}$) \cite{Zhang_TWC_2026,Du_COMST_2026}.
Performance is evaluated via radar probability of detection $P_{\rm D} = Q_1 (\sqrt{2 \gamma^{\rm s}}, \sqrt{-2 \ln P_{\rm FA}})$, with Marcum $Q$-function $Q_1(\cdot,\cdot)$ and false-alarm probability $P_{\rm FA}$, and user Gray-coded QPSK bit error rate. The threshold $R_u^{\min}$ is a minimum service rate for a coded cell-edge link rather than for an uncoded data link, so the uncoded BER below serves only to compare the interference levels of the designs.

The AO reference starts from the baseline and alternates exact MVDR combiner updates with augmented-Lagrangian gradient steps on the powers and beams ($25$ outer iterations of $8$ steps), with the multipliers shared across the realizations of the same $\omega$.
It penalizes the raw sensing violation rather than the achievability gap, so it is an uncalibrated reference rather than a strict upper bound.

Since the untrained PI-UnSL reproduces the baseline, the separation between the two curves in Figs.~\ref{fig:methods}--\ref{fig:beamforming} is the gain from training.
In Fig.~\ref{fig:methods}(a), PI-UnSL lies above the baseline frontier except at the most sensing-heavy priorities and delivers more rate than TDM at matched $P_{\rm D}$ over most of the priority range. Fig.~\ref{fig:methods}(b) shows that this gain peaks at balanced priorities and reverses only in the most sensing-heavy regime, where sharing the block brings little. The price is interference: TDM keeps an uncoded QPSK BER of about $6\times10^{-3}$, against $0.07$--$0.26$ for PI-UnSL, which accepts mutual interference over the whole block for a higher total rate.

Fig.~\ref{fig:tradeoff} explains this frontier. The baseline meets both QoS requirements together only at intermediate priorities, whereas PI-UnSL does so over most of the priority range and raises the joint-QoS probability substantially, most of all in the sensing-heavy regime. Fig.~\ref{fig:beamforming} shows the mechanism on one realization. All designs keep the main lobe on the target, but the baseline and TDM leave some user directions several dB lower because of equal RZF power sharing, whereas PI-UnSL lifts the pattern toward those users.

\label{subsec:architectures}

In Fig.~\ref{fig:architectures}, all NNs share the loss, dual, measured floor, optimizer, schedule, and checkpoint rule, over three seeds. The blank-slate networks see only the raw channels and the system constants of $\vec{q}_{\rm sys}$, while the prior-equipped ones also carry the physics conditioning. All five architectures fail comparably from scratch and perform alike once initialized and conditioned. A CNN with the prior alone recovers most of the joint QoS, but at a utility below the baseline, which the physics conditioning restores, so the physical knowledge, not the architecture, makes the problem learnable. Neither this CNN nor any blank-slate network reaches the $0.85$ validation target, so their checkpoints maximize validation joint QoS.
Time division attains a higher mean utility than PI-UnSL but a clearly lower joint QoS, and AO a higher utility and rate but a lower joint QoS over most priorities, at many times the latency. As AO is not optimal, its margin understates the utility the single pass forfeits.

\label{subsec:outdated}

Under channel aging \eqref{eq:gauss_markov}, designs computed at block $t$ are scored at block $t+n$ (Fig.~\ref{fig:outdated}). At $\omega=0.45$, PI-UnSL loses mean utility slightly faster than the baseline, falling below it from the first lag, but keeps a higher joint-QoS probability at every lag, as it does on average over the $19$ evaluated priorities. Its reliability advantage therefore survives channel staleness, although most of the margin is lost within the first block, and from five blocks on PI-UnSL falls slightly below the baseline at the most sensing-heavy priorities, $\omega\le0.25$.

\label{subsec:mobility}

With the allocation recomputed every block, Fig.~\ref{fig:mobility} isolates the geometry over $250$ continuous tracks, each running through four heading segments of $1$~s ($10^4$ blocks, so nodes move up to $30$~m per segment), scored on $20$ blocks per segment.
Both metrics rise as the nodes approach and fall as they recede, and they change little in the tangential control segment, so range rather than speed drives performance. PI-UnSL keeps its joint-QoS advantage on every segment.

\label{subsec:scalability}

\begin{table}[!tb]
	\centering
	\caption{Per-realization Complexity Order}
	\label{tab:complexity}
	\footnotesize
	\begin{tabular}{@{}ll@{}}
		\hline
		Method & Order \\
		\hline
		PI-UnSL (one pass) & baseline $+\,O(P)$, $P$ fixed by the architecture \\
		Baseline & $O\big(U^3+U^2N_{\rm t}+UN_{\rm t}^2+N_{\rm r}^3$ \\
		& $\quad+\,L_{\rm cl}N_{\rm r}N_{\rm t}^2+L_{\rm cl}N_{\rm t}N_{\rm r}^2\big)$ \\
		Time division & same as the baseline \\
		Iterative solver & baseline $+\,O\big(I_{\rm out}N_{\rm r}^3+I_{\rm out}I_{\rm in}(U^2N_{\rm t}+UN_{\rm t}^2$ \\
		& $\quad+\,L_{\rm cl}N_{\rm r}N_{\rm t}^2+L_{\rm cl}N_{\rm t}N_{\rm r}^2)\big)$ \\
		\hline
	\end{tabular}
\end{table}

PI-UnSL is slower than the closed-form baseline, mainly because of its forward pass, but it needs no iterations. Beyond the baseline it builds as its prior, its cost is fixed by the architecture, whereas the AO cost scales with $I_{\rm out}I_{\rm in}$ (Table~\ref{tab:complexity}). The timing of PI-UnSL includes its prior and physics statistics and excludes only the per-block basis, which all designs share. Because the architectural constant grows with $N_{\rm t}$ and $L_{\rm cl}$, the margin narrows for larger arrays.

\section{Conclusion}
\label{sec:conclusion}

We considered joint beamforming and power allocation for a mobile multiuser MIMO-ISAC network with clutter, residual self-interference, and outdated CSI under per-realization rate and sensing constraints. A closed-form sensing SINR at the interference-aware combiner showed that the problem is intermittently infeasible with an analytic violation floor, against which we calibrated the label-free dual to keep its multipliers bounded. The proposed PI-UnSL refines the analytical design through exact normalization and closed-form interference statistics in one pass. Its advantage is per-realization reliability rather than mean throughput: it raises joint QoS over the closed-form baseline, delivers more rate than time-division ISAC at matched detection over most of the tradeoff, and exceeds the reliability of an uncalibrated iterative solver at a fraction of its latency, under mobility and, with a reduced margin, under outdated CSI. More broadly, the physical knowledge given to the network, rather than its architecture, made this problem learnable without labels, as five architectures failed alike from scratch and performed alike with it.

\appendices
\section{Proof of Theorem~\ref{thm:max_sensing}}
\label{app:proof_max_sensing}
\textit{Step 1 (affine form).} Fix $(\vec{w}_{\rm s},\{\vec{w}_u\},\vec{u}_{\rm s})$ with $\|\vec{u}_{\rm s}\|=1$. Substituting $\mathbf R_x=p_{\rm s}\vec{w}_{\rm s}\vec{w}_{\rm s}^{\mathsf H}+\sum_u p_u\vec{w}_u\vec{w}_u^{\mathsf H}$ into \eqref{eq:interference_covariance} and \eqref{eq:target_channel} into \eqref{eq:sensing_sinr} gives
\begin{equation}
    \gamma^{\rm s}=\frac{\Lambda p_{\rm s}}{D},\qquad
    D=c_{\rm s}p_{\rm s}+\sum_{u=1}^{U}c_up_u+\sigma_{\rm r}^2>0,
    \label{eq:app_affine}
\end{equation}
with $\Lambda=G_{\rm p}|\alpha_{\rm T}|^2|\vec{u}_{\rm s}^{\mathsf H}\vec{a}_{\rm r}|^2|\vec{a}_{\rm t}^{\mathsf H}\vec{w}_{\rm s}|^2\ge0$ and
\begin{equation}
\begin{split}
    c_{\rm s}&=\textstyle\sum_{\ell}|\vec{u}_{\rm s}^{\mathsf H}\mathbf H_{{\rm cl},\ell}\vec{w}_{\rm s}|^2
    +\chi_{\rm SI}|\vec{u}_{\rm s}^{\mathsf H}\mathbf H_{\rm SI}\vec{w}_{\rm s}|^2\ge0,\\
    c_u&=|\vec{u}_{\rm s}^{\mathsf H}\mathbf H_{\rm T}\vec{w}_u|^2
    +\textstyle\sum_{\ell}|\vec{u}_{\rm s}^{\mathsf H}\mathbf H_{{\rm cl},\ell}\vec{w}_u|^2\\
    &\quad+\chi_{\rm SI}|\vec{u}_{\rm s}^{\mathsf H}\mathbf H_{\rm SI}\vec{w}_u|^2\ge0,
\end{split}
\end{equation}
which collect the clutter and SI driven by the sensing stream and the target-reflected, clutter, and SI leakage driven by user $u$. Communication power thus enters only the denominator.

\textit{Step 2 (power split).} Move $\Delta\in(0,p_u]$ from user $u$ to sensing at fixed total power. The new SINR $\Lambda(p_{\rm s}+\Delta)/(D+(c_{\rm s}-c_u)\Delta)$ is at least $\Lambda p_{\rm s}/D$ if and only if $(p_{\rm s}+\Delta)D\ge p_{\rm s}(D+(c_{\rm s}-c_u)\Delta)$, that is,
\begin{equation}
    \textstyle\sum_{u'}c_{u'}p_{u'}+\sigma_{\rm r}^2+c_u p_{\rm s}\ge0,
\end{equation}
which always holds and is strict since $\sigma_{\rm r}^2>0$. Moreover, $\partial\gamma^{\rm s}/\partial p_{\rm s}=\Lambda(\sum_u c_up_u+\sigma_{\rm r}^2)/D^2\ge0$, so any unused budget should also go to sensing. Repeating the transfer over all users yields $p_{\rm s}=P_{\max}$ and $p_u=0$, for every choice of beams and combiner.

\textit{Step 3 (combiner).} At this split, $\mathbf R_x=P_{\max}\vec{w}_{\rm s}\vec{w}_{\rm s}^{\mathsf H}$, so $D=\vec{u}_{\rm s}^{\mathsf H}\mathbf R_{\rm in}^{({\rm s})}\vec{u}_{\rm s}$ with $\mathbf R_{\rm in}^{({\rm s})}\succeq\sigma_{\rm r}^2\mathbf I_{N_{\rm r}}\succ0$, and
\begin{equation}
    \gamma^{\rm s}=G_{\rm p}P_{\max}|\alpha_{\rm T}|^2|\vec{a}_{\rm t}^{\mathsf H}\vec{w}_{\rm s}|^2\,
    \frac{|\vec{u}_{\rm s}^{\mathsf H}\vec{a}_{\rm r}|^2}{\vec{u}_{\rm s}^{\mathsf H}\mathbf R_{\rm in}^{({\rm s})}\vec{u}_{\rm s}}.
    \label{eq:app_rayleigh}
\end{equation}
With $\vec{t}=(\mathbf R_{\rm in}^{({\rm s})})^{1/2}\vec{u}_{\rm s}$, the Cauchy--Schwarz inequality gives $|\vec{u}_{\rm s}^{\mathsf H}\vec{a}_{\rm r}|^2=|\vec{t}^{\mathsf H}(\mathbf R_{\rm in}^{({\rm s})})^{-1/2}\vec{a}_{\rm r}|^2\le\|\vec{t}\|^2\,\vec{a}_{\rm r}^{\mathsf H}(\mathbf R_{\rm in}^{({\rm s})})^{-1}\vec{a}_{\rm r}$, with equality if and only if $\vec{u}_{\rm s}\propto(\mathbf R_{\rm in}^{({\rm s})})^{-1}\vec{a}_{\rm r}$~\cite[Lemma~B.9]{Bjornson_2017}. The quotient in \eqref{eq:app_rayleigh} is therefore at most $\vec{a}_{\rm r}^{\mathsf H}(\mathbf R_{\rm in}^{({\rm s})})^{-1}\vec{a}_{\rm r}$, attained by the MVDR combiner, the solution of $\min_{\vec{u}}\vec{u}^{\mathsf H}\mathbf R_{\rm in}^{({\rm s})}\vec{u}$ subject to $\vec{u}^{\mathsf H}\vec{a}_{\rm r}=1$~\cite[Eqs.~(4)--(5)]{Elbir_SPM_2023} up to a scale that the quotient ignores. Since the sensing-only split is optimal for every combiner by Step~2, maximizing over the combiner afterwards preserves it.

\textit{Step 4 (beam).} For $\|\vec{w}_{\rm s}\|=1$, the Cauchy--Schwarz inequality gives $|\vec{a}_{\rm t}^{\mathsf H}\vec{w}_{\rm s}|^2\le\|\vec{a}_{\rm t}\|^2=N_{\rm t}$, with equality for $\vec{w}_{\rm s}=\vec{a}_{\rm t}(\phi_{\rm T})/\sqrt{N_{\rm t}}$. Substituting this beam and the MVDR combiner into \eqref{eq:app_rayleigh} yields \eqref{eq:gamma_s_max}. Because $\mathbf R_{\rm in}^{({\rm s})}$ depends on $\vec{w}_{\rm s}$ through the clutter and SI illumination, the matched beam maximizes the first factor of \eqref{eq:app_rayleigh} but need not maximize the quotient, so \eqref{eq:gamma_s_max} is the SINR of an achievable allocation rather than the global maximum over beams.

\textit{Noise-limited bound.} Since $\mathbf R_{\rm in}^{({\rm s})}\succeq\sigma_{\rm r}^2\mathbf I_{N_{\rm r}}$ implies $(\mathbf R_{\rm in}^{({\rm s})})^{-1}\preceq\sigma_{\rm r}^{-2}\mathbf I_{N_{\rm r}}$, it follows that $\vec{a}_{\rm r}^{\mathsf H}(\mathbf R_{\rm in}^{({\rm s})})^{-1}\vec{a}_{\rm r}\le N_{\rm r}/\sigma_{\rm r}^2$ and hence $\gamma^{\rm s,max}\le\gamma_{\rm NL}$ in \eqref{eq:noise_limited}. More generally, $D\ge\sigma_{\rm r}^2$, $p_{\rm s}\le P_{\max}$, $|\vec{u}_{\rm s}^{\mathsf H}\vec{a}_{\rm r}|^2\le N_{\rm r}$, and $|\vec{a}_{\rm t}^{\mathsf H}\vec{w}_{\rm s}|^2\le N_{\rm t}$ in \eqref{eq:app_affine} give $\gamma^{\rm s}\le\gamma_{\rm NL}$ for every allocation, beam, and combiner, which is the bound behind Corollary~\ref{cor:range}.
\qed

\let\oldthebibliography\thebibliography
\renewcommand{\thebibliography}[1]{%
  \oldthebibliography{#1}%
  \setlength{\itemsep}{0pt}\setlength{\parsep}{0pt}}
\bibliographystyle{IEEEtran}
\bibliography{references}

\end{document}